\documentclass[journal,twoside,10pt]{IEEEtran}
\usepackage{amsmath}
\usepackage{mathrsfs}
\usepackage{amsfonts}
\usepackage[numbers,sort&compress]{natbib}
\usepackage{graphicx,color,overpic}
\usepackage{times}
\usepackage{latexsym}
\usepackage{bm}
\usepackage{amssymb}
\usepackage{cases}
\usepackage{array}
\usepackage{fancyhdr}
\usepackage{setspace}

\usepackage[caption=false,font=footnotesize,labelfont=rm,textfont=rm]{subfig}
\usepackage{subfig}
\usepackage{url}
\usepackage{algorithm}
\usepackage{algorithmic}
\usepackage{multirow}
\allowdisplaybreaks[4]

\usepackage{citesort}
\usepackage{epsfig}
\usepackage{fancybox}
\usepackage{textcomp}
\usepackage{psfrag}
\usepackage{makecell}

\newtheorem{theorem}{Theorem}

\newtheorem{remark}{Remark}

    \def\Complex{{\rm\rule[.23ex]{.03em}{1.1ex}\kern-.3em{C}}}

    \newcommand{\be}{\begin{equation}} \newcommand{\ee}{\end{equation}}
    \newcommand{\bea}{\begin{eqnarray}} \newcommand{\eea}{\end{eqnarray}}
    \newcommand{\benum}{\begin{enumerate}} \newcommand{\eenum}{\end{enumerate}}

    \newcommand{\qa}{{\bf a}}

    \newcommand{\qA}{{\bf A}}

\begin{document}

\title{Non-Square UPA-Enabled XL-MIMO Systems: Anisotropic Near-Field Characterization, Fundamental Limits, and Channel Estimation}
\author{Yilong Liu,~\IEEEmembership{Graduate Student Member,~IEEE},
Xi Yang,~\IEEEmembership{Member,~IEEE},
Jing Xu,

Jun Zhang,~\IEEEmembership{Senior Member,~IEEE},
and Shi Jin,~\IEEEmembership{Fellow,~IEEE}
\thanks{Yilong Liu, Xi Yang, and Jing Xu are with the School of Information and Electronic Engineering, East China Normal University, Shanghai 200241, China (e-mail: yilongliu@stu.ecnu.edu.cn; xyang@cee.ecnu.edu.cn; jxu@ce.ecnu.edu.cn).}
\thanks{Jun Zhang is with the Jiangsu Key Laboratory of Wireless Communications, Nanjing University of Posts and Telecommunications, Nanjing 210003, China (e-mail: zhangjun@njupt.edu.cn).}
\thanks{Shi Jin is with the National Mobile Communications Research Laboratory, Southeast University, Nanjing 210096, China (e-mail: jinshi@seu.edu.cn).}
}
\maketitle

\begin{abstract}
Extremely large-scale multiple-input multiple-output (XL-MIMO) has emerged as a promising technology for enabling next-generation communication systems.
In practice, the deployment of non-square uniform planar arrays (UPAs) fundamentally alters electromagnetic wave propagation characteristics and induces anisotropic beamfocusing capability along different axes due to the aperture disparity originating from the non-square array geometry.
To fully uncover the performance impact of such non-square geometry and thus unleash the potential of the non-square UPAs, we investigate the anisotropic near-field characteristics, fundamental limits, and channel estimation for non-square UPA-enabled XL-MIMO systems.
First, we derive the effective beamfocusing Rayleigh distances for the long and short axes of the array.
Interestingly, the radiation space of a non-square UPA can be partitioned into three regions, i.e., the fully near-field, the anisotropic near-field, and the far-field regions, and the anisotropic region asymptotically dominates the overall near-field space as the array aspect ratio increases.
Then, the asymptotic effective degree of freedom for non-square UPA-enabled XL-MIMO systems is provided, which reveals that distance-domain multiplexing is governed by the long-axis aperture in the large array aspect ratio regime.
Furthermore, the closed-form Cram\'er-Rao bound (CRB) for distance estimation and the three-dimensional (3D) position error bound (PEB) are derived to reveal the geometry-induced performance trade-offs among distance, azimuth, and elevation estimation, based on which the optimal array aspect ratio that minimizes the 3D PEB is determined.
Finally, by exploiting the anisotropic wavefront properties, we design a 3D anisotropic near-field codebook to facilitate low-complexity channel estimation for non-square UPAs.
Numerical results validate that the proposed codebook achieves comparable estimation accuracy to the 3D polar-domain codebook with substantially lower computational complexity.
\end{abstract}

\begin{IEEEkeywords}
    XL-MIMO, near-field, non-square UPA, CRB, channel estimation.
\end{IEEEkeywords}

\section{Introduction}
Extremely large-scale multiple-input multiple-output (XL-MIMO) has emerged as a pivotal technology for sixth-generation wireless networks to meet the stringent requirements of future applications \cite{Lu-24COMST}.
By deploying extremely large aperture arrays (ELAAs), such as uniform linear arrays (ULAs) \cite{XYang-24WCL}, \cite{SLiu-25JSAC} and uniform planar arrays (UPAs) \cite{HLu-22TWC}, \cite{YLiu-26TVT} with hundreds to thousands of antennas, XL-MIMO significantly enhances spatial resolution and spectral efficiency, while facilitating high-precision positioning for sensing scenarios \cite{ZWang-24COMST}.

With the substantial expansion of physical apertures in ELAAs, the channel modeling transitions from the far-field planar wavefront assumption to the near-field spherical wavefront modeling \cite{JAn-24WC}.
Enabled by near-field spherical wavefronts, XL-MIMO systems can achieve beamfocusing, concentrating electromagnetic energy at specific spatial locations.
This property introduces additional spatial degrees of freedom (DoFs) in the distance domain, thereby enhancing spatial multiplexing \cite{ZWu-23JSAC}.
Motivated by this potential, extensive research has been dedicated to near-field XL-MIMO technologies, such as near-field channel estimation \cite{CD-22TCOMM}, \cite{XShi-25TCOM}, beamforming \cite{ZWang-26TWC}, and resource allocation \cite{Bokai-24WC}.
Note that the boundary between the near-field and far-field regions is demarcated by the Rayleigh distance \cite{Balanis-16}, which is derived based on a maximum phase error criterion between the spherical wavefront and its planar approximation.
However, recent studies demonstrate that the Rayleigh distance overestimates the effective range of the near-field region, as the maximum phase error does not directly evaluate the degradation of beamforming gain or transmission capacity \cite{MCui-24TWC}.
To address this limitation, the effective beamfocusing Rayleigh distance (EBRD) has been introduced to characterize the range over which an array can sustain a finite $3\,\mathrm{dB}$ beam depth \cite{Hussain-25TWC}, \cite{Hussain-26TCOMM}.

In practical deployments, UPAs are widely adopted as the prevailing architecture owing to their capability to achieve high antenna integration densities within constrained physical spaces \cite{CChen-26TWC}, \cite{3GPP-38.901}.
Motivated by this adoption, recent studies have focused on the beamfocusing properties and associated performance limits of UPA-enabled near-field XL-MIMO systems.
Specifically, \cite{Kosasih-24TWC} evaluated the near-field beamfocusing and derived the $3\,\mathrm{dB}$ beam depth for uniform rectangular arrays.
The authors in \cite{XChen-26TCOM} studied the distance-focusing characteristics of sparse UPAs and introduced the radial resolution distance to characterize the effective beamfocusing range.
\cite{ZWang-25TWC} derived closed-form effective DoF (EDoF) expressions for UPAs based on the Green's function, revealing that the EDoF primarily depends on the physical aperture rather than the number of antennas.
Moreover, \cite{Hussain-26TCOMM} derived the beam depth for uniform rectangular arrays to uncover the impact of the array geometry on the beamfocusing capability.

To further adapt to specific spatial constraints and coverage requirements, XL-MIMO systems typically deploy non-square UPAs \cite{YuH-19TCOM}, \cite{XYang-26JIOT}.
For instance, non-square UPAs with a larger horizontal aperture but a relatively small vertical aperture are employed to efficiently accommodate users densely distributed in the azimuth plane \cite{3GPP-38.901}, \cite{XYang-26JIOT}.
Note that the beamfocusing capability in a specific direction is governed by the effective projected aperture of the array along that direction.
Thus, for a non-square UPA, the aperture discrepancy between the horizontal and vertical axes alters the wavefront properties, leading to anisotropic beamfocusing capability.
This geometry-induced property renders the Rayleigh distance derived solely from the maximum diagonal aperture inadequate to accurately capture the spatial radiation characteristics.
% However, the aforementioned studies on UPAs predominantly rely on this single Rayleigh distance for near-field characterization.
% Furthermore, the fundamental limits governed by such anisotropic near-field characteristics remain unexplored.
Specifically, existing near-field characterizations of UPAs based on this single Rayleigh distance do not explicitly identify the region where the long and short axes exhibit different channel characteristics.
The corresponding field demarcation and channel model are needed to exploit this anisotropy in performance analysis and application design.
To address these limitations, we investigate how the non-square UPA geometry affects near-field propagation characteristics and fundamental performance limits.
Specifically, we reveal the anisotropic beamfocusing characteristics and demarcate the corresponding anisotropic near-field region, based on which we analyze spatial multiplexing and parameter estimation performance and develop a codebook to facilitate low-complexity channel estimation.
The primary contributions are summarized as follows.
\begin{itemize}
\item First, we derive the closed-form EBRDs for the long and short axes of the non-square UPA by evaluating the normalized array gain.
Based on these boundaries, the radiation space of a non-square UPA can be demarcated into the fully near-field, the anisotropic near-field, and the far-field regions, and the corresponding channel models are also formulated.
It is shown that as the array aspect ratio increases, the anisotropic region asymptotically dominates the overall near-field space, which validates the necessity of the proposed field demarcation for non-square UPA deployments.
\item Second, based on this sophisticated channel modeling, we investigate the spatial multiplexing capabilities of the non-square UPA by deriving an asymptotic closed-form expression for the EDoF.
Our analysis reveals that the distance-domain multiplexing capacity is primarily determined by the effective physical aperture along the long axis in the large array aspect ratio regime, rather than solely by the number of antennas.
\item Third, we derive the closed-form Cram\'er-Rao bound (CRB) for distance estimation and the three-dimensional (3D) position error bound (PEB) with respect to the array aspect ratio, characterizing the geometry-induced performance trade-offs among distance, azimuth, and elevation estimation.
Based on this analysis, we determine the optimal array aspect ratio that minimizes the 3D PEB.
\item Finally, by exploiting the anisotropic wavefront characteristics, we design a 3D anisotropic near-field (3D-ANF) codebook, which utilizes the discrete fractional Fourier transform (DFrFT) to match the spatial chirp rate along the long axis, and employs the discrete Fourier transform (DFT) for the short axis.
Based on this codebook, the proposed low-complexity ANF-orthogonal matching pursuit (OMP) channel estimation algorithm achieves comparable estimation accuracy to the 3D polar-domain OMP algorithm while significantly reducing the computational complexity.
\end{itemize}

The remainder of this paper is organized as follows.
Section~\ref{s:Model} introduces the considered non-square UPA-enabled XL-MIMO system.
In Section~\ref{s:anisotropy_NF}, the anisotropic near-field characterization is uncovered, and we derive the EDoF in Section~\ref{s:EDoF}.
Then, we analyze the performance bounds for parameter estimation and 3D positioning in Section~\ref{s:Estimation_analysis}.
Section~\ref{s:ChannelEst} proposes a novel codebook design and a low-complexity channel estimation scheme, and finally we conclude this paper in Section~\ref{s:Conclusion}.

\textit{Notations:} Scalars, vectors, and matrices are denoted by non-boldface, boldface lowercase, and boldface uppercase letters, respectively.
Superscripts $(\cdot)^T$, $(\cdot)^H$, and $(\cdot)^{-1}$ denote the transpose, Hermitian transpose, and inverse, respectively.
The Kronecker product, the magnitude operator, and the ceiling operator are denoted by $\otimes$, $|\cdot|$, and $\lceil \cdot \rceil$, respectively.
$\mathbb{E}(\cdot)$, $\mathrm{tr}(\cdot)$, $\mathcal{F}^{(p)}\{\cdot\}$, and $\mathcal{F}^{(1)}\{\cdot\}$ are the statistical expectation, trace, $p$-th order DFrFT, and DFT operations, respectively.
$\mathbb{C}^{M\times N}$ and $\mathbb{R}^{M\times N}$ denote the space of $M \times N$ complex and real matrices, respectively, and $\Re(\cdot)$ extracts the real part.
$\mathrm{reshape}(\mathbf{x}, M, N)$ reshapes a vector $\mathbf{x}$ of length $MN$ into an $M \times N$ matrix by filling its columns sequentially.
$[\qA]_{(i,j)}$ denotes the $(i,j)$-th element of the matrix $\qA$.
$[\qa]_i$ and $\|\qa\|$ denote the $i$-th element and the Euclidean norm of the vector $\qa$, respectively.
$\mathbf{0}$ and $\mathbf{I}$ denote the all-zero vector and the identity matrix, respectively.
$\mathcal{CN}(\mu, \sigma^2)$ denotes a complex Gaussian distribution with mean $\mu$ and variance $\sigma^2$.

\section{System Model}\label{s:Model}
\begin{figure}[htbp]
       \centering
       \subfloat[]{
       \includegraphics[width=0.35\textwidth]{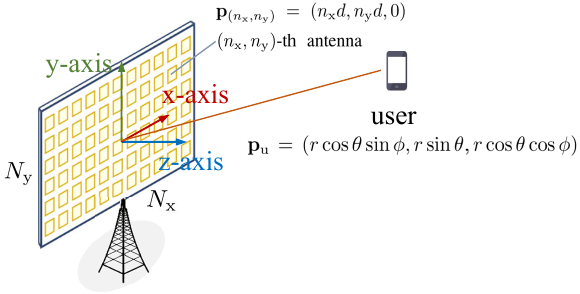}}\\
       \subfloat[]{
       \includegraphics[width=0.35\textwidth]{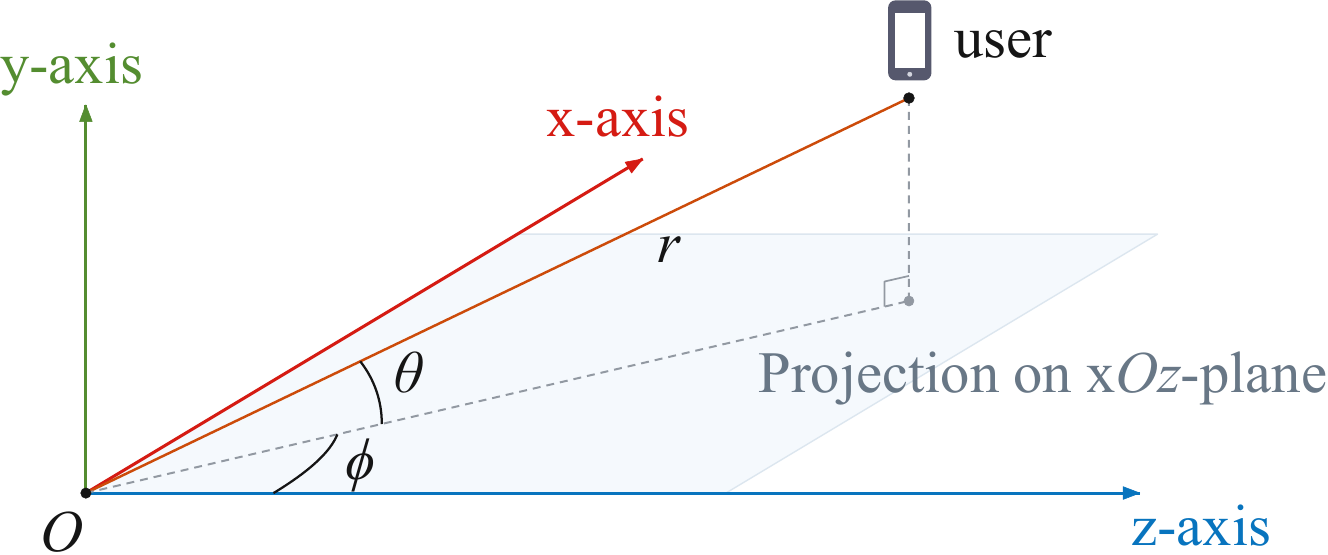}}
       \caption{The non-square UPA-enabled XL-MIMO system: (a) system layout and (b) system geometry.}\label{fig:sys_model}
\end{figure}

As depicted in Fig.\,\ref{fig:sys_model}, we consider an XL-MIMO system serving a single-antenna user, where the base station (BS) is equipped with a non-square UPA consisting of $N = N_{\mathrm{x}} \times N_{\mathrm{y}}$ antennas, with $N_{\mathrm{x}}$ and $N_{\mathrm{y}}$ denoting the numbers of antennas along the horizontal and vertical axes, respectively, and $N_{\mathrm{x}} \geq N_{\mathrm{y}}$.
The aspect ratio of the array is defined as $\gamma = \frac{N_{\mathrm{x}}}{N_{\mathrm{y}}}$, and we have $\gamma \geq 1$.
A 3D Cartesian coordinate system is established with the center of the non-square UPA as the origin, where the $x$-axis and $y$-axis align with the horizontal and vertical dimensions of the array, respectively.
Thus, the positions of the $(n_{\mathrm{x}}, n_{\mathrm{y}})$-th antenna at the BS and the user are denoted by $\mathbf{p}_{(n_{\mathrm{x}}, n_{\mathrm{y}})} = (n_{\mathrm{x}} d, n_{\mathrm{y}} d, 0)$ and $\mathbf{p}_{\mathrm{u}} = (r\cos\theta\sin\phi, r\sin\theta, r\cos\theta\cos\phi)$, respectively, where $n_{\mathrm{x}} \in \mathcal{N}_{\mathrm{x}}$ and $n_{\mathrm{y}} \in \mathcal{N}_{\mathrm{y}}$ are the antenna indices with $\mathcal{N}_{\mathrm{x}} \triangleq \{-\frac{N_{\mathrm{x}}-1}{2}, \dots, \frac{N_{\mathrm{x}}-1}{2}\}$ and $\mathcal{N}_{\mathrm{y}} \triangleq \{-\frac{N_{\mathrm{y}}-1}{2}, \dots, \frac{N_{\mathrm{y}}-1}{2}\}$,
$d = \frac{\lambda}{2}$ is the antenna spacing, and $\lambda$ denotes the carrier wavelength.
$r$ represents the distance from the user to the center of the UPA (i.e., the origin), while $\theta \in [-\frac{\pi}{2}, \frac{\pi}{2}]$ and $\phi \in [-\frac{\pi}{2}, \frac{\pi}{2}]$ are the elevation and azimuth angles-of-arrival (AoAs) for the user relative to the BS, respectively.
By applying the Fresnel approximation and omitting the cross-terms \cite{ZWu-23JSAC}, \cite{Hussain-26TCOMM}, the distance between the $(n_{\mathrm{x}}, n_{\mathrm{y}})$-th antenna at the BS and the user or the scatterer, i.e., $r_{(n_{\mathrm{x}}, n_{\mathrm{y}})}$, can be expressed as
\begin{align}\label{eq:r_n}
    r_{(n_{\mathrm{x}}, n_{\mathrm{y}})} =& \{(r\cos\theta\sin\phi - n_{\mathrm{x}} d)^2 \notag\\
    & + (r\sin\theta - n_{\mathrm{y}} d)^2 + (r\cos\theta\cos\phi)^2 \}^{\frac{1}{2}} \notag\\
    \approx & r - (n_{\mathrm{x}} d u_{\mathrm{x}} + n_{\mathrm{y}} d u_{\mathrm{y}}) \notag\\
    & + \frac{n_{\mathrm{x}}^2 d^2 (1-u_{\mathrm{x}}^2)}{2r} + \frac{n_{\mathrm{y}}^2 d^2 (1-u_{\mathrm{y}}^2)}{2r},
\end{align}
where $u_{\mathrm{x}} \triangleq \cos\theta\sin\phi$ and $u_{\mathrm{y}} \triangleq \sin\theta$.
Therefore, the near-field steering vector at the BS, i.e., $\mathbf{a}_{\mathrm{NF}}(r, u_{\mathrm{x}}, u_{\mathrm{y}}) \in \mathbb{C}^{N \times 1}$, is modeled as
\begin{align}
    \mathbf{a}_{\mathrm{NF}}(r, u_{\mathrm{x}}, u_{\mathrm{y}}) =& \frac{1}{\sqrt{N}} \left[e^{-j \frac{2\pi}{\lambda}\left(r_{(-\tilde{N}_{\mathrm{x}}, -\tilde{N}_{\mathrm{y}})}-r\right)},\cdots, \right. \notag\\
    & \qquad \left. e^{-j \frac{2\pi}{\lambda}\left(r_{(\tilde{N}_{\mathrm{x}}, \tilde{N}_{\mathrm{y}})}-r\right)}\right]^T \notag\\
    \overset{\mathrm{(a)}}{\approx} & \mathbf{a}_{\mathrm{x,NF}}(r, u_{\mathrm{x}}) \otimes \mathbf{a}_{\mathrm{y,NF}}(r, u_{\mathrm{y}}),
\end{align}
where the approximation $\mathrm{(a)}$ follows from \eqref{eq:r_n}, $\tilde{N}_{\mathrm{x}} \triangleq \frac{N_{\mathrm{x}}-1}{2}$, and $\tilde{N}_{\mathrm{y}} \triangleq \frac{N_{\mathrm{y}}-1}{2}$.
$\mathbf{a}_{\mathrm{x,NF}}(r, u_{\mathrm{x}}) \in \mathbb{C}^{N_{\mathrm{x}} \times 1}$ and $\mathbf{a}_{\mathrm{y,NF}}(r, u_{\mathrm{y}}) \in \mathbb{C}^{N_{\mathrm{y}} \times 1}$ denote the near-field steering vectors along the horizontal and vertical axes, respectively, which are given by
\begin{subequations}
    \begin{align}
        [\mathbf{a}_{\mathrm{x,NF}}(r, u_{\mathrm{x}})]_i =& \frac{1}{\sqrt{N_{\mathrm{x}}}} e^{-j \frac{2\pi}{\lambda} \left( -n_{\mathrm{x},i} d u_{\mathrm{x}} + \frac{n_{\mathrm{x},i}^2 d^2 (1-u_{\mathrm{x}}^2)}{2r} \right) }, \label{eq:ax}\\
        [\mathbf{a}_{\mathrm{y,NF}}(r, u_{\mathrm{y}})]_j =& \frac{1}{\sqrt{N_{\mathrm{y}}}} e^{-j \frac{2\pi}{\lambda} \left( -n_{\mathrm{y},j} d u_{\mathrm{y}} + \frac{n_{\mathrm{y},j}^2 d^2 (1-u_{\mathrm{y}}^2)}{2r} \right) }, \label{eq:ay}
    \end{align}
\end{subequations}
where $n_{\mathrm{x},i}$ and $n_{\mathrm{y},j}$ denote the $i$-th and $j$-th elements of the ordered sets $\mathcal{N}_{\mathrm{x}}$ and $\mathcal{N}_{\mathrm{y}}$, respectively.

The radiation space surrounding an antenna array is conventionally partitioned into the near-field and far-field regions,\footnote{Since the reactive near-field is confined within a few wavelengths from the array surface, its impact is practically negligible for typical wireless communication systems \cite{ZXu-25TWC}.
As a result, the near-field throughout this paper specifically refers to the radiative near-field.} with the Rayleigh distance widely adopted as the boundary \cite{Balanis-16}.
For a BS equipped with a UPA, existing literature generally leverages the maximum physical aperture along the main diagonal with length $D = d\sqrt{(N_{\mathrm{x}}-1)^2 + (N_{\mathrm{y}}-1)^2}$ to characterize the corresponding Rayleigh distance $r_{\mathrm{Rayleigh}} = \frac{2D^2}{\lambda}$.
However, the aperture discrepancy between the two axes of a non-square UPA inevitably induces anisotropic beamfocusing capability.
Consequently, relying on the Rayleigh distance derived solely from the maximum diagonal aperture cannot accurately capture the spatial radiation characteristics.
To address this limitation, we evaluate the EBRDs of non-square UPAs along the long and short axes and uncover the anisotropic near-field region in the next section.

\section{Anisotropic Near-Field Characterization}\label{s:anisotropy_NF}
In the near-field region, beamforming techniques enable the BS to achieve beamfocusing, which concentrates the electromagnetic energy within a confined spatial region, manifesting as a finite beamwidth and beam depth \cite{JAn-24WC}.
Unlike the isotropic beamfocusing of square arrays, the disparity in physical apertures between the long and short axes of the non-square UPA implies that the short axis loses its beamfocusing capability at a much shorter distance than the long axis.
Such a phenomenon potentially gives rise to a unique region where the user resides in the near-field of the long axis while appearing in the far-field of the short axis.
Fundamental questions thus arise: \textit{Does such a region physically exist and how does the array aspect ratio govern its boundaries?}
To address these, we analyze the beamfocusing mechanism of the non-square UPA and reveal the spatial partitioning of the anisotropic near-field.

To characterize the beamfocusing capability of the non-square UPA, the normalized array gain can be defined as $g(r_{\mathrm{F}}, u_{\mathrm{x}}, u_{\mathrm{y}})$, as follows
\begin{align}\label{eq:g0}
    g(r_{\mathrm{F}}, u_{\mathrm{x}}, u_{\mathrm{y}}) \triangleq \left| \mathbf{w}^H(r_{\mathrm{F}}, u_{\mathrm{x}}, u_{\mathrm{y}}) \mathbf{a}_{\mathrm{NF}}(r, u_{\mathrm{x}}, u_{\mathrm{y}}) \right|^2,
\end{align}
where $\mathbf{w}(r_{\mathrm{F}}, u_{\mathrm{x}}, u_{\mathrm{y}}) \in \mathbb{C}^{N \times 1}$ is the beamforming vector.
Without loss of generality, we assume $\mathbf{w}(r_{\mathrm{F}}, u_{\mathrm{x}}, u_{\mathrm{y}}) = \mathbf{a}_{\mathrm{NF}}(r_{\mathrm{F}}, u_{\mathrm{x}}, u_{\mathrm{y}})$.
By leveraging the mixed-product property of the Kronecker product, i.e., $(\mathbf{A} \otimes \mathbf{B})^H (\mathbf{C} \otimes \mathbf{D}) = (\mathbf{A}^H \mathbf{C}) \otimes (\mathbf{B}^H \mathbf{D})$ for arbitrary matrices $\mathbf{A}$, $\mathbf{B}$, $\mathbf{C}$, and $\mathbf{D}$, \eqref{eq:g0} can be rewritten as
\begin{align}\label{eq:gain}
    g(r_{\mathrm{F}}, u_{\mathrm{x}}, u_{\mathrm{y}}) =& \left| \mathbf{a}_{\mathrm{x,NF}}^H(r_{\mathrm{F}}, u_{\mathrm{x}}) \mathbf{a}_{\mathrm{x,NF}}(r, u_{\mathrm{x}}) \right|^2 \notag\\
    & \times \left| \mathbf{a}_{\mathrm{y,NF}}^H(r_{\mathrm{F}}, u_{\mathrm{y}}) \mathbf{a}_{\mathrm{y,NF}}(r, u_{\mathrm{y}}) \right|^2 \notag\\
    =& g_{\mathrm{x}}(r_{\mathrm{F}}, u_{\mathrm{x}}) g_{\mathrm{y}}(r_{\mathrm{F}}, u_{\mathrm{y}}),
\end{align}
where $g_{\mathrm{x}}(r_{\mathrm{F}}, u_{\mathrm{x}}) \triangleq | \mathbf{a}_{\mathrm{x,NF}}^H(r_{\mathrm{F}}, u_{\mathrm{x}}) \mathbf{a}_{\mathrm{x,NF}}(r, u_{\mathrm{x}}) |^2$ and $g_{\mathrm{y}}(r_{\mathrm{F}}, u_{\mathrm{y}}) \triangleq | \mathbf{a}_{\mathrm{y,NF}}^H(r_{\mathrm{F}}, u_{\mathrm{y}}) \mathbf{a}_{\mathrm{y,NF}}(r, u_{\mathrm{y}}) |^2$ denote the normalized array gains corresponding to the long and short axes, respectively.
As observed from \eqref{eq:gain}, the normalized array gain of the UPA can be decoupled into the product of the individual gains along the $x$ (long) and $y$ (short) axes.
This decoupling property implies that the spatial radiation characteristics and the corresponding EBRDs can be evaluated independently along the $x$- and $y$-axes.
Such a decoupled evaluation is crucial for analyzing non-square UPAs, where the aperture disparity between the long and short axes dictates distinct beamfocusing capabilities.
Substituting \eqref{eq:ax} into $g_{\mathrm{x}}(r_{\mathrm{F}}, u_{\mathrm{x}})$ and applying the continuous aperture approximation yields
\begin{align}\label{eq:gain_Fresnel_x}
    g_{\mathrm{x}}(r_{\mathrm{F}}, u_{\mathrm{x}}) =& \left| \frac{1}{N_{\mathrm{x}}} \sum_{n_{\mathrm{x}}} e^{ j \frac{2\pi}{\lambda} \frac{n_{\mathrm{x}}^2 d^2 (1-u_{\mathrm{x}}^2)}{2} \left( \frac{1}{r_{\mathrm{F}}} - \frac{1}{r} \right) } \right|^2 \notag\\
    \approx & \left| \int_{-1/2}^{1/2} e^{ j \frac{\pi N_{\mathrm{x}}^2 d^2 (1-u_{\mathrm{x}}^2)}{\lambda} z_{\mathrm{eff}} t^2 } dt \right|^2 \notag\\
    =& \frac{C^2(\eta_{\mathrm{x}}) + S^2(\eta_{\mathrm{x}})}{\eta_{\mathrm{x}}^2},
\end{align}
where $z_{\mathrm{eff}} \triangleq | \frac{1}{r_{\mathrm{F}}} - \frac{1}{r}|$, $t \triangleq \frac{n_{\mathrm{x}}}{N_{\mathrm{x}}}$, and $\eta_{\mathrm{x}} \triangleq \sqrt{\frac{N_{\mathrm{x}}^2 d^2 (1-u_{\mathrm{x}}^2)}{2\lambda}z_{\mathrm{eff}}}$.
The terms $C(\eta)$ and $S(\eta)$ denote the Fresnel integral functions, defined as $C(\eta) = \int_0^\eta \cos(\frac{\pi}{2}v^2)dv$ and $S(\eta) = \int_0^\eta \sin(\frac{\pi}{2}v^2)dv$.
Similarly, $g_{\mathrm{y}}(r_{\mathrm{F}}, u_{\mathrm{y}})$ can be obtained as
\begin{align}\label{eq:gain_Fresnel_y}
    g_{\mathrm{y}}(r_{\mathrm{F}}, u_{\mathrm{y}}) \approx \frac{C^2(\eta_{\mathrm{y}}) + S^2(\eta_{\mathrm{y}})}{\eta_{\mathrm{y}}^2},
\end{align}
where $\eta_{\mathrm{y}} \triangleq \sqrt{\frac{N_{\mathrm{y}}^2 d^2 (1-u_{\mathrm{y}}^2)}{2\lambda}z_{\mathrm{eff}}}$.
Comparing the Maclaurin expansions of the discrete summations with the continuous-integral expressions in \eqref{eq:gain_Fresnel_x} and \eqref{eq:gain_Fresnel_y} shows that the corresponding gain errors are $O(\eta_{\mathrm{x}}^4)$ and $O(\eta_{\mathrm{y}}^4)$, respectively.
Since both $\eta_{\mathrm{x}}$ and $\eta_{\mathrm{y}}$ are much smaller than $1$ when $r$ is near the focal distance $r_{\mathrm{F}}$, the continuous-aperture approximations are accurate even when the corresponding axis contains only a few elements.
It is observed that the normalized array gain $g(r_{\mathrm{F}}, u_{\mathrm{x}}, u_{\mathrm{y}})$ can be decoupled into the product of $g_{\mathrm{x}}(r_{\mathrm{F}}, u_{\mathrm{x}})$ and $g_{\mathrm{y}}(r_{\mathrm{F}}, u_{\mathrm{y}})$.

\begin{figure}[htbp]
       \centering
       \includegraphics[width=0.35\textwidth]{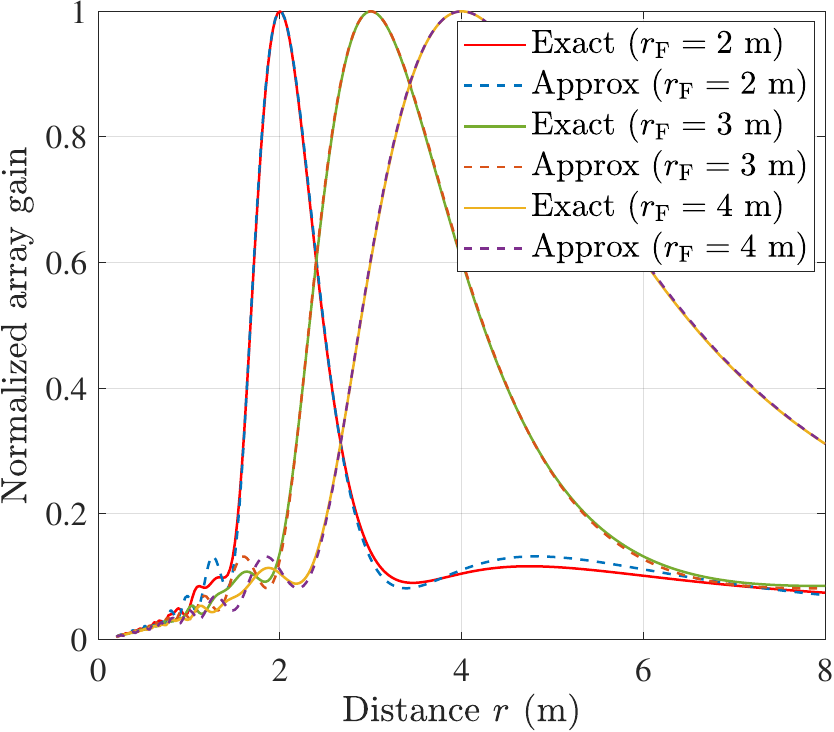}
       \caption{Normalized array gain versus the observation distance $r$ for different focal distances, i.e., $r_{\mathrm{F}} = 2.0$, $3.0$, and $4.0\,\mathrm{m}$.}\label{fig:Gain_r}
\end{figure}

Fig.\,\ref{fig:Gain_r} illustrates the normalized array gain versus the observation distance for different focal distances, where $N_{\mathrm{x}} = 128$, $N_{\mathrm{y}} = 16$, $\theta = 30^{\circ}$, $\phi = 30^{\circ}$, and the carrier frequency is $28,\mathrm{GHz}$.
As observed, the normalized array gain exhibits distinct peaks at the focal distances and attenuates elsewhere, demonstrating that the non-square UPA achieves near-field beamfocusing in the distance domain.
Furthermore, the normalized array gain curves evaluated using the approximated propagation distance closely align with those obtained via the exact spatial distance in the main-lobe regions.
This alignment indicates that the Fresnel approximation and the omission of cross terms introduce only minor deviations in the main-lobe gain curves (i.e., a maximum relative error below $1.1\%$ over the sampled points within the $3\,\mathrm{dB}$ main lobes).
These results, thus, further support the decoupling of $g(r_{\mathrm{F}}, u_{\mathrm{x}}, u_{\mathrm{y}})$ in \eqref{eq:gain} into the product of $g_{\mathrm{x}}(r_{\mathrm{F}}, u_{\mathrm{x}})$ and $g_{\mathrm{y}}(r_{\mathrm{F}}, u_{\mathrm{y}})$, which are directly determined by the aperture sizes of the $x$- and $y$-axes, respectively.
For $\gamma > 1$, the discrepancy in aperture sizes between the two axes yields distinct beamfocusing capabilities.
Specifically, the long axis sustains effective beamfocusing at a substantially greater distance than the short axis.
This disparity uncovers an anisotropic effective beamfocusing region.
To quantitatively characterize this anisotropic beamfocusing property, we establish the following theorem.

\begin{theorem}\label{the:Ry}
    When the focal distance $r_{\mathrm{F}}$ reaches or exceeds a boundary $R_{\mathrm{y}}$, the short axis loses its beamfocusing capability, i.e., its $3\,\mathrm{dB}$ beam depth $r_{\mathrm{BD,y}}$ is infinite.
    This boundary is defined as the EBRD of the short axis, which is given by
    \begin{align}\label{eq:Ry}
        R_{\mathrm{y}} = \frac{N_{\mathrm{y}}^2 d^2 (1-u_{\mathrm{y}}^2)}{2\lambda \eta_0^2},
    \end{align}
    where $\eta_0$ represents the primary root of the half-power equation $\frac{C^2(\eta_{\mathrm{y}}) + S^2(\eta_{\mathrm{y}})}{\eta_{\mathrm{y}}^2} = 0.5$.
\end{theorem}

\begin{IEEEproof}
See Appendix~\ref{proof:the:Ry}.
\end{IEEEproof}

Similarly, the EBRD of the long axis, denoted by $R_{\mathrm{x}}$, can be expressed as
\begin{align}\label{eq:Rx}
    R_{\mathrm{x}} = \frac{N_{\mathrm{x}}^2 d^2 (1-u_{\mathrm{x}}^2)}{2\lambda \eta_0^2}.
\end{align}

\begin{figure}[htbp]
       \centering
       \subfloat[]{
       \includegraphics[width=0.235\textwidth]{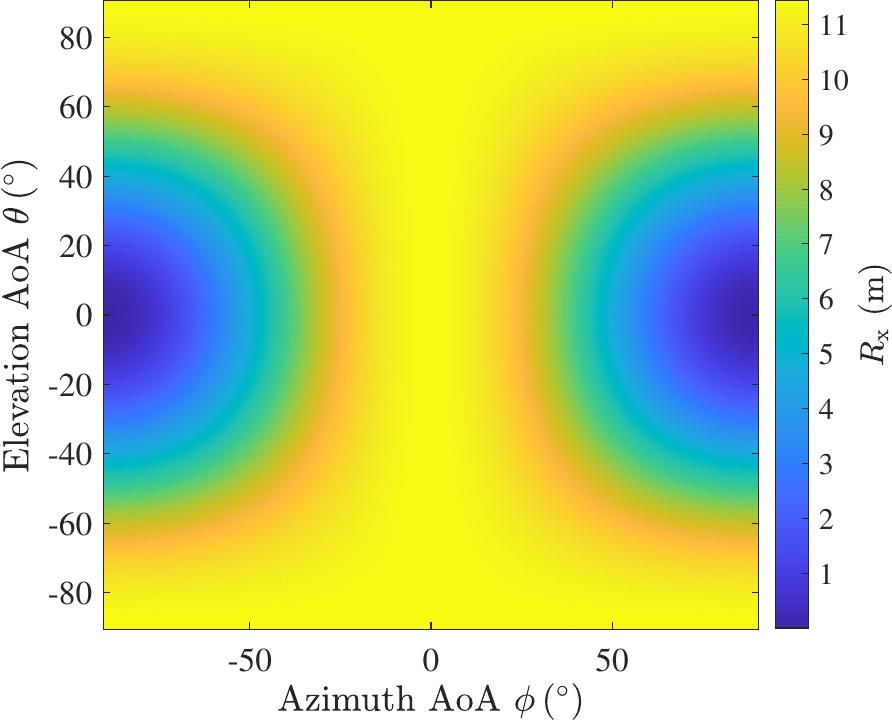}}
       \subfloat[]{
       \includegraphics[width=0.235\textwidth]{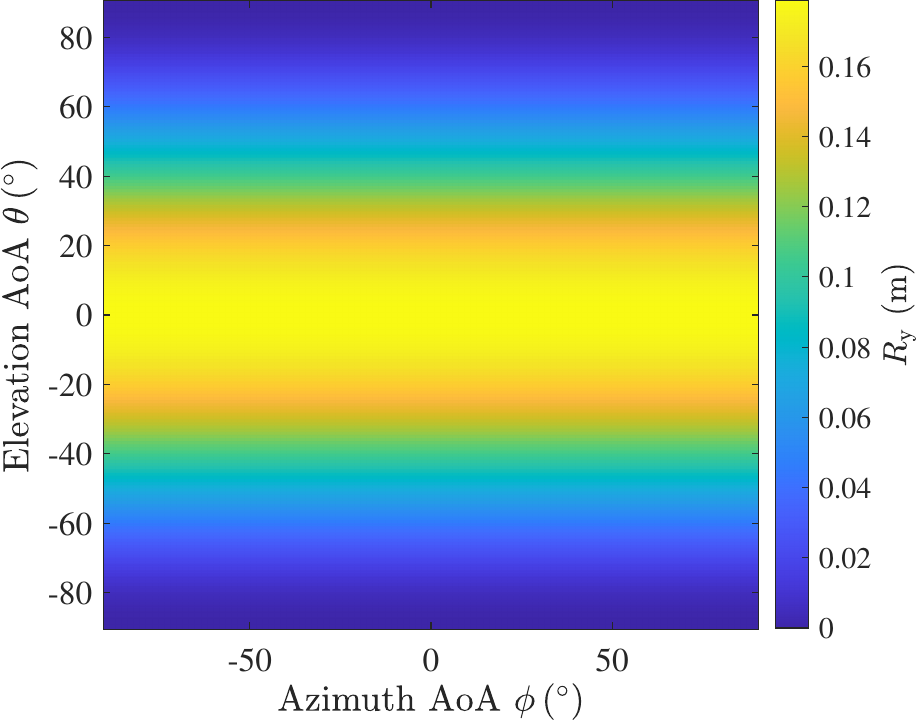}}\\
       \caption{The spatial distribution of the EBRDs for the long and short axes (i.e., $R_{\mathrm{x}}$ and $R_{\mathrm{y}}$) in the 2D angular domain.}\label{fig:Boundary_hotmap}
\end{figure}

Fig.\,\ref{fig:Boundary_hotmap} depicts the spatial distribution of the EBRDs for the long and short axes (i.e., $R_{\mathrm{x}}$ and $R_{\mathrm{y}}$) in the two-dimensional (2D) angular domain, with $N_{\mathrm{x}} = 128$, $N_{\mathrm{y}} = 16$, and a carrier frequency of $28\,\mathrm{GHz}$.
It is observed that both $R_{\mathrm{x}}$ and $R_{\mathrm{y}}$ are highly sensitive to the AoAs.
Specifically, $R_{\mathrm{y}}$ depends solely on the elevation AoA, peaking at $0^{\circ}$ and decreasing monotonically as the deviation increases.
In contrast, $R_{\mathrm{x}}$ is jointly dictated by both the elevation and azimuth AoAs, decreasing rapidly when the observation direction deviates from the broadside.
This phenomenon occurs because the beamfocusing capability in a specific direction is governed by the effective projected aperture in that direction.
Furthermore, across the majority of the 2D angular domain, $R_{\mathrm{x}}$ is substantially larger than $R_{\mathrm{y}}$, highlighting the disparity in beamfocusing capabilities between the two axes of the non-square UPA.\footnote{Note that the condition $R_{\mathrm{y}} > R_{\mathrm{x}}$ can potentially hold when $\theta \rightarrow 0^{\circ}$ and $\phi \rightarrow \pm 90^{\circ}$.
Since this angular range falls outside the typical service sectors of cellular networks \cite{3GPP-38.901}, it is omitted from our analysis.}
To precisely capture the anisotropic beamfocusing properties of the non-square UPA, we establish the following remark.

\begin{remark}
    For a given observation direction $(u_{\mathrm{x}}, u_{\mathrm{y}})$, the radiation field of a non-square UPA can be partitioned into three distinct regions: the fully near-field region ($r < R_{\mathrm{y}}$), the anisotropic near-field region ($R_{\mathrm{y}} \leq r < R_{\mathrm{x}}$), and the far-field region ($r \geq R_{\mathrm{x}}$).
\end{remark}

Note that, different from the Rayleigh distance which is defined based on a wavefront phase-error criterion, we capitalize on the proposed boundaries $R_{\mathrm{y}}$ and $R_{\mathrm{x}}$ that characterize the beamfocusing capability to demarcate the radiation regions for the non-square UPA.
In the fully near-field region ($r < R_{\mathrm{y}}$), the beamfocusing of the non-square UPA is jointly governed by both the long and short axes, with the steering vectors of both axes modeled by spherical wavefronts.
In the anisotropic near-field region ($R_{\mathrm{y}} \leq r < R_{\mathrm{x}}$), the short axis can no longer sustain its beamfocusing capability, leaving the beamfocusing capability of the array exclusively governed by the long axis.
Within this region, the long-axis steering vector maintains a spherical wavefront model, whereas the short-axis steering vector can be approximated by a planar wavefront model, and thus the steering vector of the non-square UPA, i.e., $\mathbf{a}_{\mathrm{ANF}}(r, u_{\mathrm{x}}, u_{\mathrm{y}}) \in \mathbb{C}^{N \times 1}$, can be formulated as
\begin{align}\label{eq:a_ANF}
    \mathbf{a}_{\mathrm{ANF}}(r, u_{\mathrm{x}}, u_{\mathrm{y}}) = \mathbf{a}_{\mathrm{x,NF}}(r, u_{\mathrm{x}}) \otimes \mathbf{a}_{\mathrm{y,FF}}(u_{\mathrm{y}}),
\end{align}
where $\mathbf{a}_{\mathrm{y,FF}}(u_{\mathrm{y}}) \in \mathbb{C}^{N_{\mathrm{y}} \times 1}$ denotes the short-axis far-field steering vector, given by
\begin{align}
    [\mathbf{a}_{\mathrm{y,FF}}(u_{\mathrm{y}})]_j = \frac{1}{\sqrt{N_{\mathrm{y}}}} e^{j \frac{2\pi}{\lambda} n_{\mathrm{y},j} d u_{\mathrm{y}} }.
\end{align}
In the far-field region ($r \geq R_{\mathrm{x}}$), neither the long axis nor the short axis can achieve beamfocusing, causing their individual steering vectors to degenerate into planar wavefront models.
Therefore, the far-field steering vector of the non-square UPA, i.e., $\mathbf{a}_{\mathrm{FF}}(u_{\mathrm{x}}, u_{\mathrm{y}}) \in \mathbb{C}^{N \times 1}$, is written as
\begin{align}
    \mathbf{a}_{\mathrm{FF}}(u_{\mathrm{x}}, u_{\mathrm{y}}) = \mathbf{a}_{\mathrm{x,FF}}(u_{\mathrm{x}}) \otimes \mathbf{a}_{\mathrm{y,FF}}(u_{\mathrm{y}}),
\end{align}
where $\mathbf{a}_{\mathrm{x,FF}}(u_{\mathrm{x}}) \in \mathbb{C}^{N_{\mathrm{x}} \times 1}$ represents the long-axis far-field steering vector, defined as
\begin{align}
    [\mathbf{a}_{\mathrm{x,FF}}(u_{\mathrm{x}})]_i = \frac{1}{\sqrt{N_{\mathrm{x}}}} e^{j \frac{2\pi}{\lambda} n_{\mathrm{x},i} d u_{\mathrm{x}} }.
\end{align}

\begin{figure}[htbp]
       \centering
       \subfloat[]{
       \includegraphics[width=0.35\textwidth]{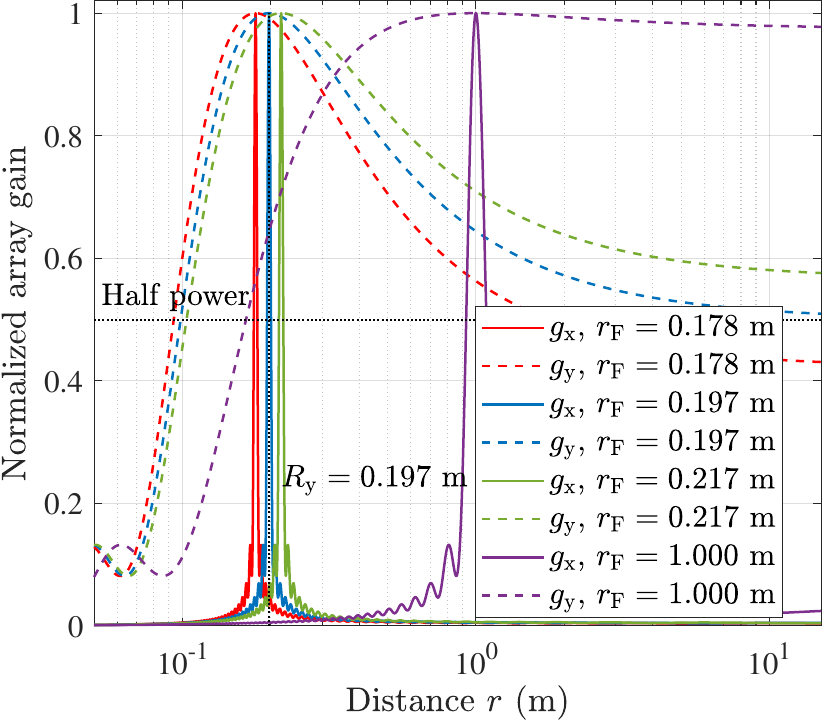}}\\
       \subfloat[]{
       \includegraphics[width=0.35\textwidth]{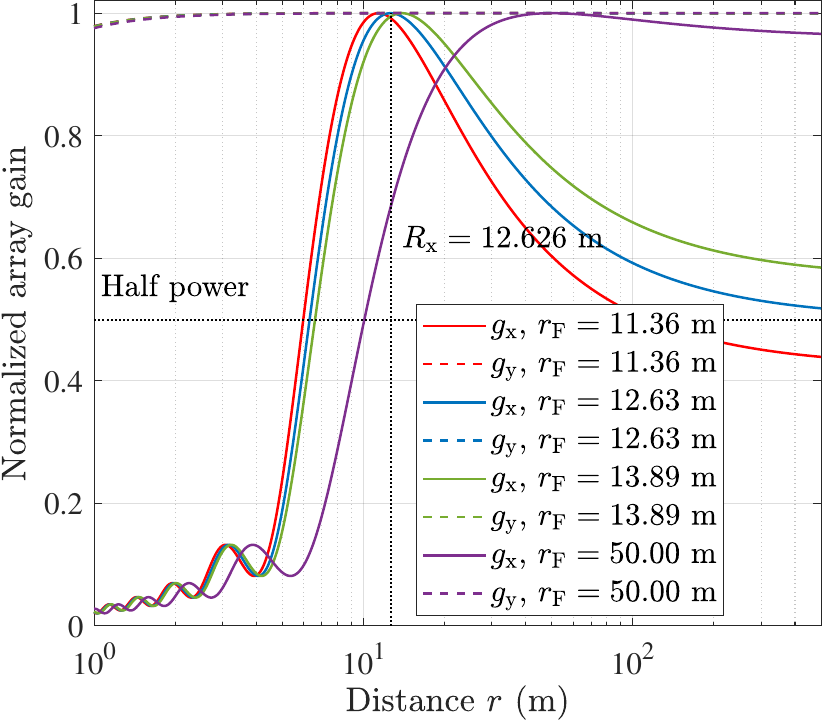}}
       \caption{The normalized array gains of the long axis and the short axis versus the observation distance $r$ along the broadside direction under various focal distances.}\label{fig:Gain_xy}
\end{figure}

To evaluate the beamfocusing performance, Fig.\,\ref{fig:Gain_xy} plots the normalized array gains of the long axis and the short axis versus the observation distance along the broadside direction (i.e., $\theta = 0^{\circ}$ and $\phi = 0^{\circ}$) under various focal distances, where $N_{\mathrm{x}} = 128$, $N_{\mathrm{y}} = 16$, and the carrier frequency is $28\,\mathrm{GHz}$.
As shown in Fig.\,\ref{fig:Gain_xy}(a), even when the focal distance approaches $R_{\mathrm{y}}$ ($r_{\mathrm{F}} = 0.9R_{\mathrm{y}} \approx 0.178\,\mathrm{m}$), the normalized array gains of both axes exhibit sharp peaks at the focal distance, featuring narrow $3\,\mathrm{dB}$ beam depths.
This confirms that both axes possess beamfocusing capabilities.
When the focal distance resides in the anisotropic near-field region ($R_{\mathrm{y}} \leq r_{\mathrm{F}} < R_{\mathrm{x}}$), the long axis maintains finite-depth beamfocusing even near either boundary, e.g., at $r_{\mathrm{F}} = 1.1R_{\mathrm{y}} \approx 0.217\,\mathrm{m}$ and $r_{\mathrm{F}} = 0.9R_{\mathrm{x}} \approx 11.36\,\mathrm{m}$.
However, the short-axis normalized array gain fails to form a sharp peak (i.e., its $3\,\mathrm{dB}$ beam depth extends to infinity), rendering the short axis incapable of effectively confining the radiated energy within a finite distance interval.
Consequently, the anisotropic channel model captures this difference in the beamfocusing capabilities of the two axes even near the two region boundaries.
Furthermore, as depicted in Fig.\,\ref{fig:Gain_xy}(b), when the focal distance reaches or exceeds $R_{\mathrm{x}}$ ($r_{\mathrm{F}} \geq R_{\mathrm{x}}$, e.g., $r_{\mathrm{F}} = 1.1R_{\mathrm{x}} \approx 13.89\,\mathrm{m}$ and $r_{\mathrm{F}} = 50\,\mathrm{m}$), the long axis likewise fails to sustain finite-depth beamfocusing.
At this stage, the non-square UPA completely lacks beamfocusing capability, causing the steering vectors of both axes to degenerate into planar wavefront models.
In summary, the observed beamfocusing evolution corroborates the validity of partitioning the radiation field into the fully near-field, the anisotropic near-field, and the far-field regions using $R_{\mathrm{y}}$ and $R_{\mathrm{x}}$ as the boundaries.

\begin{figure}[htbp]
       \centering
       \includegraphics[width=0.35\textwidth]{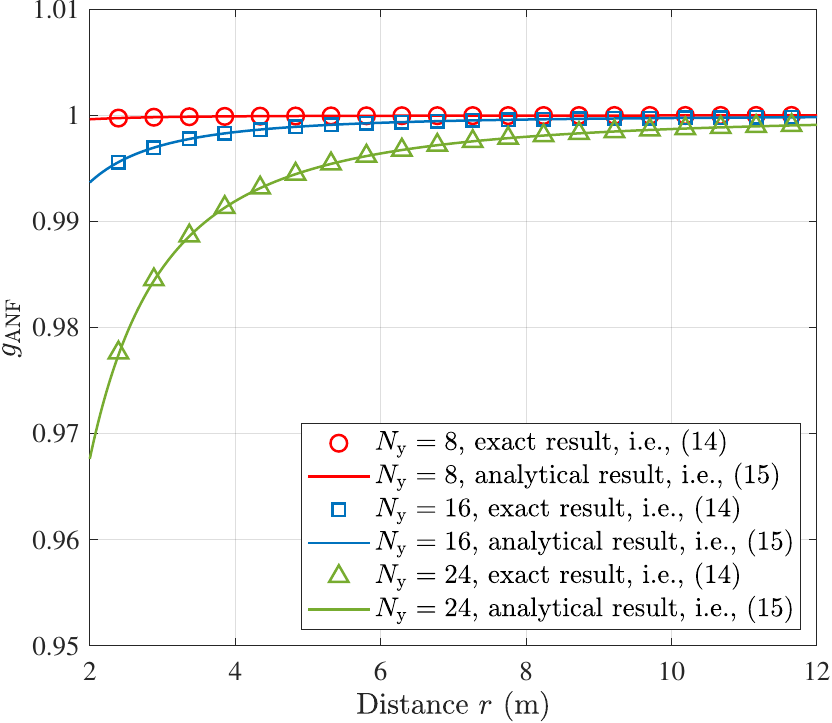}
       \caption{Normalized beamforming gain versus the user distance for different numbers of antennas along the short axis.}\label{fig:Gain_loss}
\end{figure}

Since the field demarcation is based on beamfocusing capability \cite{Hussain-25TWC}, \cite{Hussain-26TCOMM}, \cite{Kosasih-24TWC}, normalized beamforming gain is utilized to assess the accuracy of the anisotropic near-field channel model, i.e., \eqref{eq:a_ANF}, which can be defined as $g_{\mathrm{ANF}} \triangleq |\mathbf{a}_{\mathrm{ANF}}^H(r, u_{\mathrm{x}}, u_{\mathrm{y}})\mathbf{a}_{\mathrm{NF}}(r, u_{\mathrm{x}}, u_{\mathrm{y}})|^2$.
Thus, $1 - g_{\mathrm{ANF}}$ quantifies the power-gain loss caused by neglecting the short-axis quadratic phase.
As both steering vectors are normalized and share the same long-axis component in \eqref{eq:ax}, we obtain
\begin{align}
    g_{\mathrm{ANF}} \overset{\mathrm{(b)}}{=}& | \mathbf{a}_{\mathrm{y,FF}}^H(r, u_{\mathrm{y}})\mathbf{a}_{\mathrm{y,NF}}(r, u_{\mathrm{y}}) |^2 \notag\\
    =& \left| \frac{1}{N_{\mathrm{y}}} \sum_{n_y} e^{-j\frac{\pi d^2(1-u_{\mathrm{y}}^2)}{\lambda r}n_y^2}\right|^2\notag\\
    =& \left| \frac{1}{N_{\mathrm{y}}} \sum_{n_{\mathrm{y}}} e^{ -j \frac{2\pi \tilde{\eta}_{\mathrm{y}}^2 n_{\mathrm{y}}^2}{N_{\mathrm{y}}^2} } \right|^2,
\end{align}
where the equality $\mathrm{(b)}$ follows from the mixed-product property of the Kronecker product, and $\tilde{\eta}_{\mathrm{y}} \triangleq \sqrt{\frac{N_{\mathrm{y}}^2 d^2 (1-u_{\mathrm{y}}^2)}{2\lambda r}}$.
Applying Euler's formula and the Maclaurin expansion to the cosine and sine terms, we obtain
\begin{align}
    g_{\mathrm{ANF}} \approx 1 - \frac{\pi^2\tilde{\eta}_{\mathrm{y}}^4}{45}\frac{(N_{\mathrm{y}}^2-1)(N_{\mathrm{y}}^2-4)}{N_{\mathrm{y}}^4}.
\end{align}
Note that $\tilde{\eta}_{\mathrm{y}} < 1$ generally holds in the considered anisotropic near-field region, and thus its fourth power is much smaller than $1$.
In addition, $\frac{\pi^2}{45}$ is less than $1$, and $\frac{(N_{\mathrm{y}}^2-1)(N_{\mathrm{y}}^2-4)}{N_{\mathrm{y}}^4}$ can be approximated as $1$.
As a result, we obtain $g_{\mathrm{ANF}} \approx 1$, i.e., the beamforming vector constructed from \eqref{eq:a_ANF} achieves nearly the same power gain as that obtained using the near-field channel model.
Fig.\,\ref{fig:Gain_loss} illustrates the normalized beamforming gain versus the user distance for $N_{\mathrm{y}} \in \{8, 16, 24\}$, where $N_{\mathrm{x}} = 128$, the carrier frequency is $28\,\mathrm{GHz}$, $d=\frac{\lambda}{2}$, and $\theta = \phi = 0^\circ$.
As the distance increases, the normalized gain approaches $1$.
For a fixed distance, a smaller $N_{\mathrm{y}}$ yields a higher $g_{\mathrm{ANF}}$ because the smaller short-axis aperture reduces the phase mismatch introduced by \eqref{eq:a_ANF}.
Over $r \in [2, 12]\,\mathrm{m}$, the power-gain loss is below $1\%$ for $N_{\mathrm{y}} \in \{8, 16\}$ and below $4\%$ for $N_{\mathrm{y}} = 24$, validating its accuracy under the considered non-square UPA-enabled XL-MIMO systems.

To quantify the spatial extent of the anisotropic near-field region, we define $K(\gamma)$ as the normalized error between the EBRDs of the long axis and the entire array, and $\bar{K}(\gamma)$ as the depth ratio of the anisotropic near-field region relative to the entire near-field region, which are formulated as
\begin{subequations}
\begin{align}
    K(\gamma) \triangleq & \frac{R_{\mathrm{array}} - R_{\mathrm{x}}}{R_{\mathrm{array}}},\\
    \bar{K}(\gamma) \triangleq & \frac{R_{\mathrm{x}} - R_{\mathrm{y}}}{R_{\mathrm{array}}},
\end{align}
\end{subequations}
where $R_{\mathrm{array}}$ represents the EBRD of the entire array \cite{Hussain-26TCOMM}.
Accordingly, we establish the following theorem.

\begin{theorem}\label{the:K}
    Along the broadside direction, $K(\gamma)$ asymptotically decays at a rate of $\mathcal{O}(\gamma^{-4})$ as $\gamma \to \infty$, yielding
    \begin{align}
        K(\gamma) \approx \frac{\pi^2\eta_0^3}{45|F'(\eta_0)|\gamma^4},
    \end{align}
    where $F'(v)$ denotes the first derivative of the gain function $F(v) \triangleq \frac{C^2(v) + S^2(v)}{v^2}$. In addition, as $\gamma \to \infty$, $\bar{K}(\gamma)$ asymptotically approaches
    \begin{align}
        \bar{K}(\gamma) \approx 1 - \frac{1}{\gamma^2}.
    \end{align}
\end{theorem}

\begin{IEEEproof}
    See Appendix~\ref{proof:the:K}.
\end{IEEEproof}

{\it Theorem~\ref{the:K}} reveals two fundamental physical properties regarding the near-field region of the non-square UPA.
On the one hand, the normalized error $K(\gamma)$ rapidly vanishes at a rate of $\mathcal{O}(\gamma^{-4})$, implying that for practical deployments (e.g., $\gamma = 4$ \cite{3GPP-38.901}), the EBRD is predominantly governed by the long axis, i.e., $R_{\mathrm{array}} \approx R_{\mathrm{x}}$.
On the other hand, the depth ratio $\bar{K}(\gamma)$ asymptotically approaches $1 - \frac{1}{\gamma^2}$, demonstrating that as the array geometry increasingly deviates from a square configuration, the anisotropic near-field region dominates the entire near-field region.

\begin{figure}[htbp]
       \centering
       \subfloat[]{
       \includegraphics[width=0.235\textwidth]{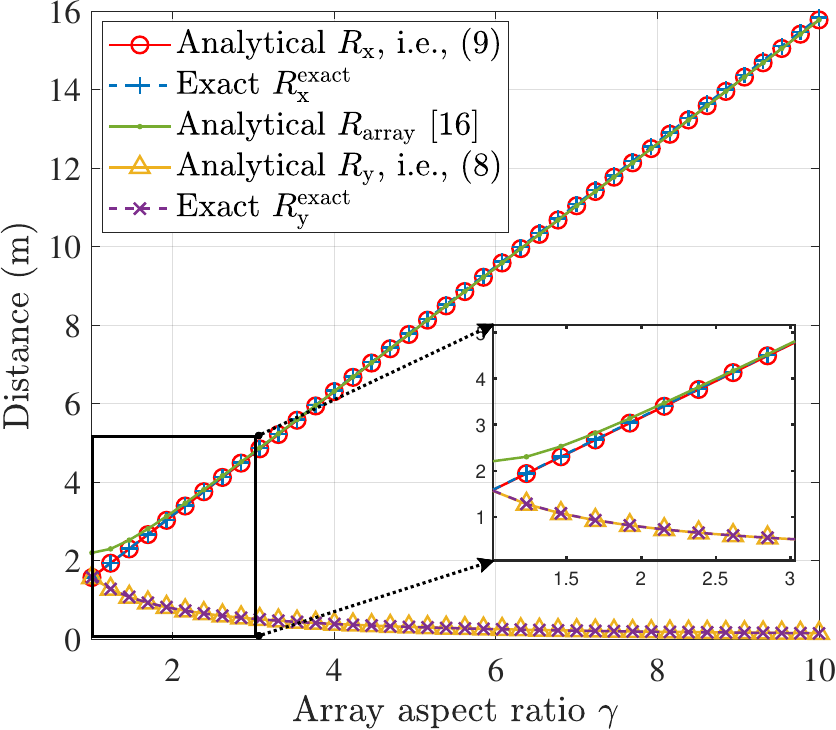}}
       \subfloat[]{
       \includegraphics[width=0.235\textwidth]{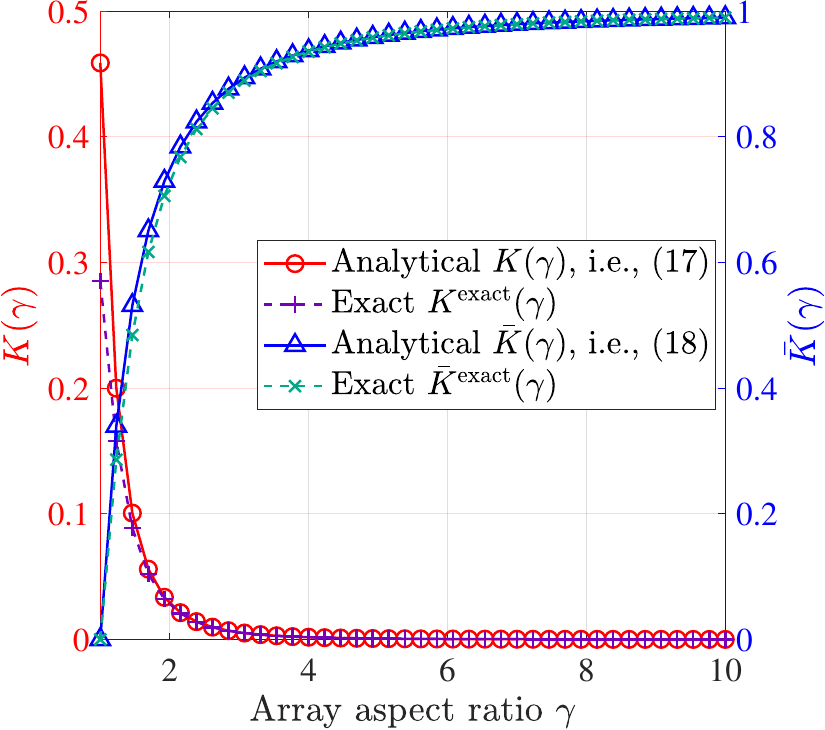}}\\
       \caption{(a) The EBRD along the broadside direction versus the array aspect ratio $\gamma$, (b) the normalized error $K(\gamma)$ and the depth ratio $\bar{K}(\gamma)$ versus the array aspect ratio $\gamma$.}\label{fig:Boundary}
\end{figure}

Fig.\,\ref{fig:Boundary}(a) plots the EBRDs along the broadside direction versus the array aspect ratio, where $N = 2048$ and the carrier frequency is $28\,\mathrm{GHz}$.
To analyze the effect of the array aspect ratio, we relax $N_{\mathrm{x}}$ and $N_{\mathrm{y}}$ to be real-valued \cite{Kosasih-24TWC}, \cite{Bjornson-15TWC} here and in the subsequent analysis, thereby treating $\gamma$ as a continuous variable. For a practical array, the aspect ratio should be determined by comparing the performance of integer configurations satisfying $N_{\mathrm{x}}N_{\mathrm{y}} = N$.
It is observed that the derived $R_{\mathrm{x}}$ and $R_{\mathrm{y}}$ perfectly align with their exact values (i.e., $R_{\mathrm{x}}^{\mathrm{exact}}$ and $R_{\mathrm{y}}^{\mathrm{exact}}$), verifying the accuracy of the proposed beamfocusing analysis.
When $\gamma$ approaches $1$ (i.e., the array tends toward a square configuration), the long axis and the short axis possess comparable physical apertures, contributing equally to the overall beamfocusing capability.
In this case, driven jointly by both axes, the entire array can sustain beamfocusing over a greater distance than the single long axis, i.e., $R_{\mathrm{array}}$ is larger than $R_{\mathrm{x}}$.
However, as $\gamma$ increases, $R_{\mathrm{y}}$ rapidly diminishes while $R_{\mathrm{x}}$ increases sharply, demonstrating that the beamfocusing capability of the short axis is restricted to a limited distance, whereas the long axis can sustain it over a substantially extended distance.
This disparity leads to an expansive anisotropic near-field region for the non-square UPA.
Within this region (e.g., $\gamma > 2$), the beamfocusing capability of the array is predominantly governed by the long axis, while the impact of the short axis becomes negligible, causing $R_{\mathrm{array}}$ to coincide with $R_{\mathrm{x}}$.

Fig.\,\ref{fig:Boundary}(b) demonstrates $K(\gamma)$ and $\bar{K}(\gamma)$ versus the array aspect ratio, where $N = 2048$ and the carrier frequency is $28\,\mathrm{GHz}$.
As depicted, the analytical results derived in {\it Theorem~\ref{the:K}} closely match the exact results across the entire observation interval and align perfectly for $\gamma > 2$, validating the theoretical derivations.
Moreover, as $\gamma$ increases, $K(\gamma)$ rapidly converges to $0$ while $\bar{K}(\gamma)$ approaches $1$, uncovering that for highly non-square UPAs, the long axis determines the EBRD for the array, and the anisotropic near-field region occupies almost the entire near-field region.

\section{Analysis of EDoF}\label{s:EDoF}
% Spatial degrees of freedom (DoF) characterize the number of independent sub-channels capable of supporting parallel transmissions in a wireless channel.
% In practice, however, the eigenvalues of the spatial correlation matrix, i.e., the power gains of these individual sub-channels, typically decay rapidly beyond a certain index.
% As a result, the effective DoF (EDoF) is adopted to quantify the number of dominant sub-channels that offer practical contributions \cite{ZWang-25TWC}.
Although spatial DoFs characterize the number of independent sub-channels, the rapid decay of spatial correlation eigenvalues necessitates the adoption of EDoF to quantify the dominant sub-channels with practical power gains \cite{ZWang-25TWC}.
For non-square UPA-enabled XL-MIMO systems, the anisotropic beamfocusing capability reshapes these sub-channels, implying that the near-field multiplexing capacity can be manipulated through array geometric reconfiguration.
Within the anisotropic near-field region identified in Section~\ref{s:anisotropy_NF}, the short axis loses its beamfocusing capability while the long axis retains it.
We derive an asymptotic closed-form expression for the EDoF to quantify the distance-domain spatial multiplexing capability retained in this region.

The EDoF is defined as the effective rank of the spatial correlation matrix $\mathbf{R} \in \mathbb{C}^{N \times N}$, which is expressed as \cite{ZWang-25TWC}
\begin{align}\label{eq:EDoF0}
    \mathrm{EDoF} \triangleq \frac{\{\mathrm{tr}(\mathbf{R})\}^2}{\mathrm{tr}(\mathbf{R}^2)}.
\end{align}
For users in the same direction, the quadratic phase at each antenna varies linearly with $\frac{1}{r}$.
Thus, the same distance separation produces smaller quadratic phase differences at longer distances, where user channels may remain highly correlated \cite{ZWu-23JSAC}.
We consider that the user is uniformly distributed in the inverse-distance domain $\zeta \in [\zeta_{\min}, \zeta_{\max}]$ with $\zeta \triangleq \frac{1}{r}$, and the probability density function of the inverse distance is calculated as $f(\zeta) = \frac{1}{\zeta_{\max} - \zeta_{\min}}$.
Therefore, users are distributed more densely at short ranges, where the quadratic phase varies rapidly with distance, and more sparsely at long ranges, where it varies slowly.
The resulting EDoF provides a benchmark for assessing the distance-domain multiplexing potential of the anisotropic near-field channel.
Accordingly, $\mathbf{R}$ can be formulated as \cite{ZDong-22CL}
\begin{align}\label{eq:R_intP}
    \mathbf{R} =& \int_{0}^{\zeta_{\mathrm{x}}} \mathbf{a}_{\mathrm{FF}}(u_{\mathrm{x}}, u_{\mathrm{y}}) \mathbf{a}_{\mathrm{FF}}^H(u_{\mathrm{x}}, u_{\mathrm{y}}) \frac{1}{\zeta_{\max}} d\zeta \notag\\
    &+ \int_{\zeta_{\mathrm{x}}}^{\zeta_{\mathrm{y}}} \mathbf{a}_{\mathrm{ANF}}\left(\frac{1}{\zeta}, u_{\mathrm{x}}, u_{\mathrm{y}}\right) \mathbf{a}_{\mathrm{ANF}}^H\left(\frac{1}{\zeta}, u_{\mathrm{x}}, u_{\mathrm{y}}\right) \frac{1}{\zeta_{\max}} d\zeta \notag\\
    &+ \int_{\zeta_{\mathrm{y}}}^{\zeta_{\max}} \mathbf{a}_{\mathrm{NF}}\left(\frac{1}{\zeta}, u_{\mathrm{x}}, u_{\mathrm{y}}\right) \mathbf{a}_{\mathrm{NF}}^H\left(\frac{1}{\zeta}, u_{\mathrm{x}}, u_{\mathrm{y}}\right) \frac{1}{\zeta_{\max}} d\zeta,
\end{align}
where $\zeta_{\mathrm{x}} \triangleq \min(\zeta_{\max}, \frac{1}{R_{\mathrm{x}}})$ and $\zeta_{\mathrm{y}} \triangleq \min(\zeta_{\max}, \frac{1}{R_{\mathrm{y}}})$.
To evaluate the near-field spatial multiplexing performance of the non-square UPA, we introduce the following theorem.

\begin{theorem}\label{the:EDoF}
    In the large array aspect ratio regime, the asymptotic closed-form expression for the EDoF of the non-square UPA along the broadside direction is given by
    \begin{align}\label{eq:EDoF}
        \mathrm{EDoF} \approx \frac{\xi}{\ln \xi + \gamma_{\mathrm{E}} + \ln(2\pi) - 1},
    \end{align}
    where $\xi \triangleq \frac{N \gamma d^2(\zeta_{\mathrm{y}}-\zeta_{\mathrm{x}})}{2\lambda}$ and $\gamma_{\mathrm{E}}$ denotes the Euler-Mascheroni constant.
\end{theorem}

\begin{IEEEproof}
    See Appendix~\ref{proof:the:EDoF}.
\end{IEEEproof}

The EDoF in \eqref{eq:EDoF} elucidates the relationship between the distance-domain spatial multiplexing capability of the non-square UPA and its geometry.
Note that $\gamma N = N^2_{\mathrm{x}}$ and therefore \eqref{eq:EDoF} demonstrates that the multiple-access capacity is primarily governed by the effective physical aperture along the long axis rather than solely the number of antennas.
For a fixed $N$, a larger $\gamma$ increases the long-axis aperture and hence the distance-domain EDoF, consistent with the established trend for elongated UPAs \cite{Hussain-26TCOMM}.
As a consequence, for practical system design, deploying a horizontally elongated non-square UPA can exploit the distance-domain multiplexing potential and enlarge the spatial multiplexing capacity through geometric reconfiguration, without incurring additional hardware overhead.
Nevertheless, such an increase in $\gamma$ also introduces a fundamental trade-off between the distance-domain multiplexing capability and the elevation angular resolution.
When $\gamma = N$, the array degenerates into a ULA and loses its spatial resolution in the elevation domain, and meanwhile the EDoF achieves its maximum value, i.e., $\mathrm{EDoF}_{\max} \approx \frac{\tilde{\xi}}{\ln (\tilde{\xi}) + \gamma_{\mathrm{E}} + \ln(2\pi) - 1}$ with $\tilde{\xi} = \frac{N^2 d^2(\zeta_{\mathrm{y}}-\zeta_{\mathrm{x}})}{2\lambda}$.

\begin{figure}[htbp]
       \centering
       \includegraphics[width=0.35\textwidth]{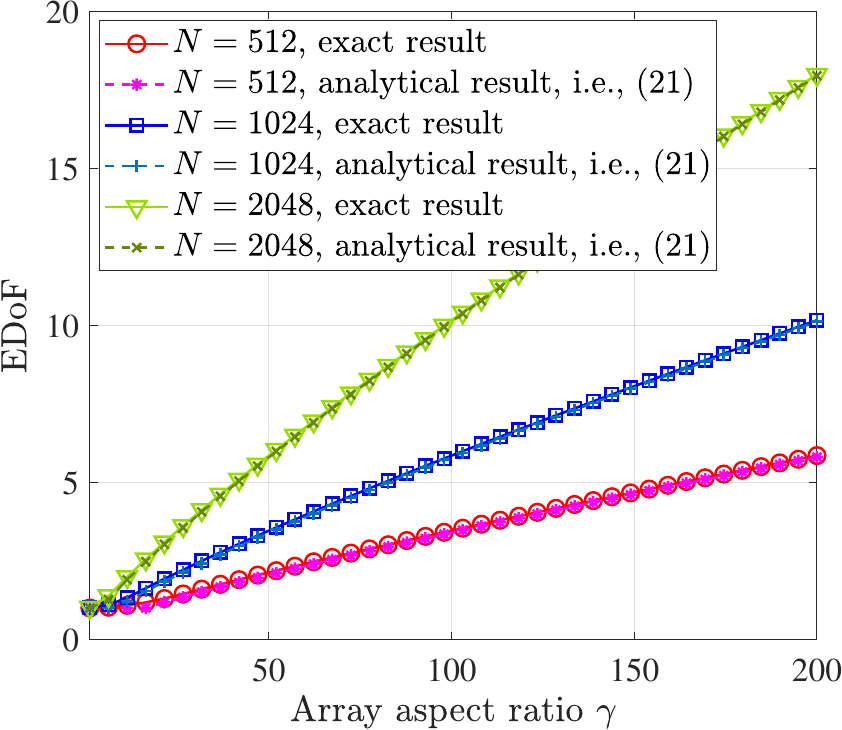}
       \caption{The EDoF along the broadside direction versus the array aspect ratio $\gamma$ under different numbers of antennas, i.e., $N = 512, 1024$, and $2048$.}\label{fig:EDoF_gamma}
\end{figure}

Fig.\,\ref{fig:EDoF_gamma} illustrates the EDoF along the broadside direction versus the array aspect ratio under different numbers of antennas.
It is observed that the analytical results match the exact results, even when the array aspect ratio is not very large (e.g., $\gamma = 6$), validating the theoretical analysis in {\it Theorem~\ref{the:EDoF}}.
Furthermore, the EDoF increases with both the array aspect ratio and the number of antennas.
This scaling law is attributed to the fact that increasing either $\gamma$ or $N$ expands the aperture of the array, thereby providing a finer spatial resolution and thus a larger EDoF in the distance domain.

\section{Performance Bounds for Parameter Estimation and 3D Positioning}\label{s:Estimation_analysis}
Beyond affecting spatial multiplexing, the array aspect ratio also governs the resolutions in the distance and angular domains, introducing a performance trade-off among the estimation accuracies of the distance, the azimuth AoA, and the elevation AoA.
To quantify these geometric impacts, we derive the closed-form CRB for distance estimation and the 3D PEB.
Based on these analytical bounds, we further determine the optimal array aspect ratio that minimizes the PEB.

Given a pilot signal $s$ transmitted by the single-antenna user,\footnote{Without loss of generality, we assume $s = 1$.} the received signal at the BS, i.e., $\mathbf{y}_1 \in \mathbb{C}^{N \times 1}$, can be expressed as
\begin{align}
    \mathbf{y}_1 = \alpha \sqrt{N} \mathbf{a}_{\mathrm{NF}}(r, u_{\mathrm{x}}, u_{\mathrm{y}}) s + \mathbf{n}_1,
\end{align}
where $\alpha = |\alpha| e^{j \varphi_{\alpha}}$ denotes the complex channel gain, and $\mathbf{n}_1 \sim \mathcal{CN}(\mathbf{0}, \sigma_1^2 \mathbf{I})$ represents the additive white Gaussian noise (AWGN) received at the BS.
To evaluate the estimation performance of distance for the non-square UPA, we establish the following theorem.

\begin{theorem}\label{the:CRB_r}
    The closed-form CRB for distance estimation under the non-square UPA is expressed with respect to the array aspect ratio as
    \begin{align}\label{eq:CRB_r}
        \mathrm{CRB}_{r} = \frac{90 \lambda^2 r^4}{\pi^2 d^4 \rho_1 N^3 \left\{ \frac{\gamma^2 (1-u_{\mathrm{x}}^2)^2}{1 + \frac{\gamma N d^2 u_{\mathrm{x}}^2}{15 r^2}} + \frac{\gamma^{-2} (1-u_{\mathrm{y}}^2)^2}{1 + \frac{\gamma^{-1} N d^2 u_{\mathrm{y}}^2}{15 r^2}} \right\}},
    \end{align}
    where $\rho_1 = \frac{|\alpha|^2}{\sigma_1^2}$ denotes the received signal-to-noise ratio (SNR).
    Along the broadside direction, \eqref{eq:CRB_r} simplifies to
    \begin{align}\label{eq:CRB_r0}
        \mathrm{CRB}_{r} = \frac{90 \lambda^2 r^4}{\pi^2 d^4 \rho_1 N^3 (\gamma^2 + \gamma^{-2}) }.
    \end{align}
\end{theorem}

\begin{IEEEproof}
    See Appendix~\ref{proof:the:CRB_r}.
\end{IEEEproof}

\begin{figure}[htbp]
       \centering
       \includegraphics[width=0.35\textwidth]{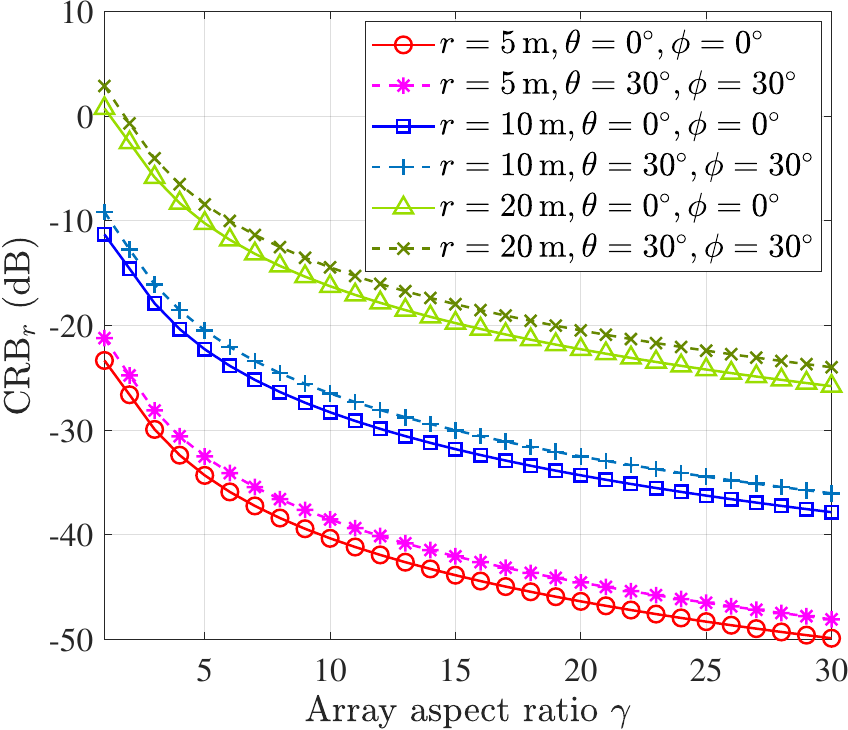}
       \caption{The CRB for distance estimation versus the array aspect ratio $\gamma$ under various focal distances and observation angles.}\label{fig:CRB_r_gamma}
\end{figure}

As indicated by \eqref{eq:CRB_r0}, $\mathrm{CRB}_{r}$ is proportional to $r^4$, which implies that the distance estimation accuracy severely degrades as the user approaches the far-field region.
For a fixed number of antennas, increasing $\gamma$ leads to a rapid decline in $\mathrm{CRB}_{r}$, thereby achieving superior distance estimation performance.
Notice that $\mathrm{CRB}_{r}$ reaches its global maximum if and only if $\gamma = 1$ (i.e., a square UPA configuration).
Fig.\,\ref{fig:CRB_r_gamma} plots the CRB for distance estimation versus the array aspect ratio under various focal distances and observation angles, where $N = 2048$ and the carrier frequency is $28\,\mathrm{GHz}$.
Evidently, whether along the broadside direction or at off-boresight angles, increasing $\gamma$ significantly reduces $\mathrm{CRB}_{r}$, validating the superiority of the non-square UPA in distance estimation.

{\it Theorem~\ref{the:CRB_r}} also validates the effectiveness of the anisotropic near-field boundaries $R_{\mathrm{y}}$ and $R_{\mathrm{x}}$ established in Section~\ref{s:anisotropy_NF}.
Specifically, \eqref{eq:J_r} reveals that the Fisher information for $r$, i.e., $J_{\mathrm{r}}$, can be decoupled into two independent components dominated by the long axis and the short axis, respectively.
When the user enters the anisotropic near-field region, the Fisher information of the short axis decays to a negligible magnitude compared to that of the long axis, resulting in $J_{\mathrm{r}}$ being governed by the long axis.
% While increasing $\gamma$ enhances the distance and azimuth AoA estimation performance of the non-square UPA, it degrades the elevation AoA estimation accuracy, introducing a fundamental performance trade-off in 3D positioning.
While increasing $\gamma$ enhances the distance and azimuth AoA estimation performance of the non-square UPA, such an increase simultaneously degrades the elevation AoA estimation accuracy, introducing a fundamental performance trade-off in 3D positioning.
To characterize the impact of $\gamma$ on the positioning performance, we introduce the PEB and establish the following theorem.

\begin{theorem}\label{the:PEB}
    The closed-form PEB for the non-square UPA, which establishes the lower bound on the root mean square positioning error, is calculated as
    \begin{align}\label{eq:PEB}
        \mathrm{PEB} =& \bigg\{\mathrm{CRB}_{r} + \frac{r^2}{u_{\mathrm{z}}^2} \left\{ (1-u_{\mathrm{y}}^2)\mathrm{CRB}_{u_{\mathrm{x}}}\right. \notag\\
        & \left. + (1-u_{\mathrm{x}}^2)\mathrm{CRB}_{u_{\mathrm{y}}} + 2u_{\mathrm{x}}u_{\mathrm{y}}\mathrm{Cov}(u_{\mathrm{x}}, u_{\mathrm{y}}) \right\}\bigg\}^{\frac{1}{2}},
    \end{align}
    where $u_{\mathrm{z}} \triangleq \sqrt{1 - u_{\mathrm{x}}^2 - u_{\mathrm{y}}^2}$, $\mathrm{CRB}_{u_{\mathrm{x}}}$ and $\mathrm{CRB}_{u_{\mathrm{y}}}$ denote the CRBs for $u_{\mathrm{x}}$ and $u_{\mathrm{y}}$, respectively, shown as \eqref{eq:CRB_xy}.
    Along the broadside direction, \eqref{eq:PEB} simplifies to
    \begin{align}\label{eq:PEB0}
        \mathrm{PEB} = \sqrt{\frac{90 \lambda^2 r^4}{\pi^2 d^4 \rho_1 N^3}\frac{1}{\gamma^2 + \gamma^{-2}} + \frac{3 \lambda^2 r^2}{2\pi^2 \rho_1 N^2 d^2} (\gamma+\gamma^{-1})}.
    \end{align}
    The optimal array aspect ratio $\gamma_{\mathrm{opt}}$ that minimizes the $\mathrm{PEB}$ in \eqref{eq:PEB0} is the unique real root of the following equation:
    \begin{align}\label{eq:gamma_opt}
        \frac{2\gamma(1+\gamma^{-2})}{(\gamma^2 + \gamma^{-2})^2} = \frac{\kappa_2}{\kappa_1},
    \end{align}
    where $\kappa_1 \triangleq \frac{90\lambda^2 r^4}{\pi^2 d^4 \rho_1 N^3}$ and $\kappa_2 \triangleq \frac{3\lambda^2 r^2}{2\pi^2\rho_1 N^2 d^2}$.
    Furthermore, in the large array aspect ratio regime, the asymptotic closed-form expression for $\gamma_{\mathrm{opt}}$ is given by
    \begin{align}\label{eq:gamma_opt_asy}
        \gamma_{\mathrm{opt}} \approx \left( \frac{120 r^2}{N d^2} \right)^{\frac{1}{3}}.
    \end{align}
\end{theorem}

\begin{IEEEproof}
    See Appendix~\ref{proof:the:PEB}.
\end{IEEEproof}

To analyze the impact of the non-square UPA geometry on the 3D positioning performance, the $\mathrm{PEB}^2$ along the broadside direction can be decoupled into a distance-domain error component proportional to $\frac{1}{\gamma^2 + \gamma^{-2}}$ and an angular-domain error component proportional to $\gamma + \gamma^{-1}$, as given in \eqref{eq:PEB0}.
It is observed that the angular-domain component achieves its minimum if and only if $\gamma = 1$ (i.e., a square UPA configuration), whereas the distance-domain component is maximized under this condition.
This demonstrates a performance trade-off in non-square UPA-enabled XL-MIMO systems, where increasing the array aspect ratio enhances the distance-domain resolution at the expense of angular-domain resolution.
Consequently, the geometry that minimizes the 3D PEB cannot be selected based on beam depth alone, and \eqref{eq:gamma_opt_asy} uncovers that the optimal array aspect ratio $\gamma_{\mathrm{opt}}$ minimizing the PEB is proportional to the cube root of the squared distance and inversely proportional to the cube root of the number of antennas, i.e., $\gamma_{\mathrm{opt}} \propto r^{\frac{2}{3}}$ and $\gamma_{\mathrm{opt}} \propto N^{-\frac{1}{3}}$.
This analytical insight provides guidelines for practical system designs.
For instance, by exploiting the property $\gamma_{\mathrm{opt}} \propto r^{\frac{2}{3}}$, the BS can dynamically activate specific sub-arrays based on the prior user position to optimize system energy efficiency.

\begin{figure}[htbp]
       \centering
       \subfloat[]{
       \includegraphics[width=0.235\textwidth]{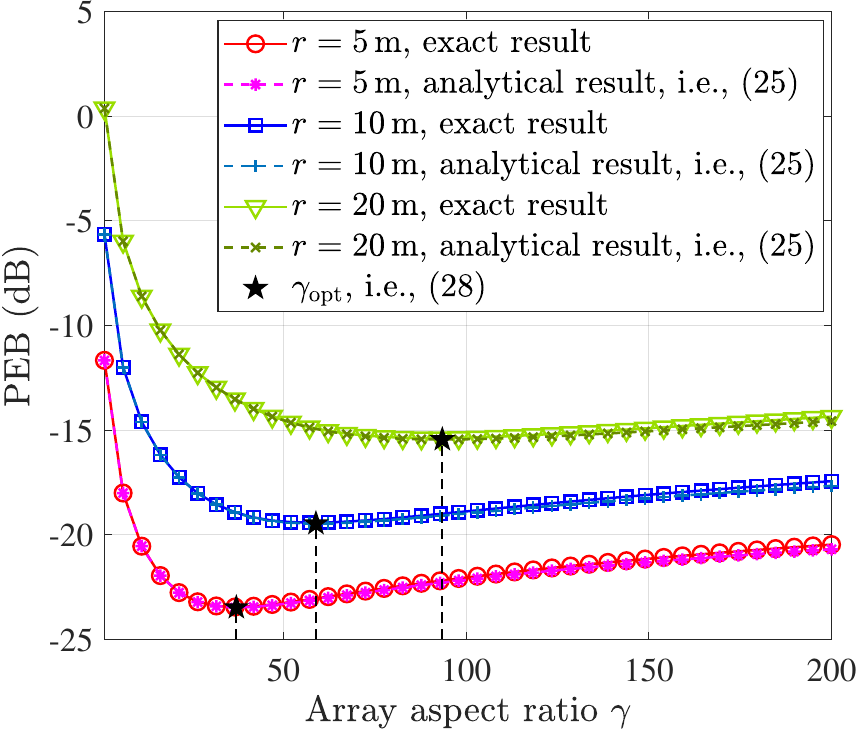}}
       \subfloat[]{
       \includegraphics[width=0.235\textwidth]{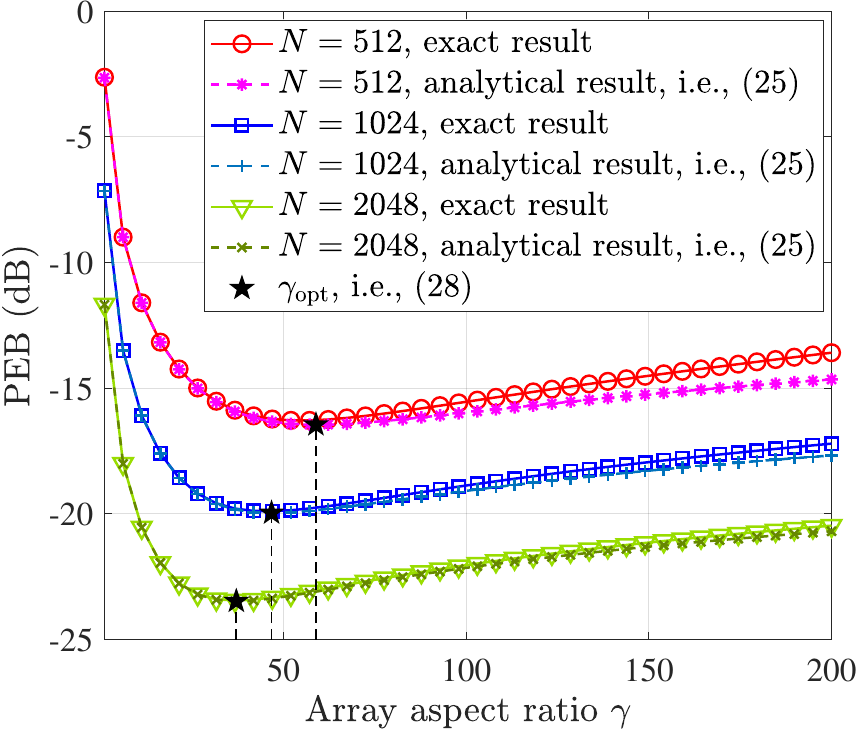}}\\
       \caption{The PEB along the broadside direction versus the array aspect ratio $\gamma$ (a) under various distances $r \in \{5, 10, 20\}\,\mathrm{m}$, and (b) under various numbers of antennas $N \in \{512, 1024, 2048\}$.}\label{fig:PEB_gamma}
\end{figure}

The PEB along the broadside direction versus the array aspect ratio is evaluated under various distances $r \in \{5, 10, 20\}\,\mathrm{m}$ in Fig.\,\ref{fig:PEB_gamma}(a) and various numbers of antennas $N \in \{512, 1024, 2048\}$ in Fig.\,\ref{fig:PEB_gamma}(b), respectively, with the carrier frequency set to $28\,\mathrm{GHz}$, $N = 2048$ for Fig.\,\ref{fig:PEB_gamma}(a), and $r = 5\,\mathrm{m}$ for Fig.\,\ref{fig:PEB_gamma}(b).
As $\gamma$ increases, the PEB first decreases and then increases, revealing a performance trade-off among the distance, the azimuth AoA, and the elevation AoA estimation for 3D positioning.
The asymptotic optimal array aspect ratio $\gamma_{\mathrm{opt}}$ derived in \eqref{eq:gamma_opt_asy} perfectly coincides with the global minimum of the PEB curves in all considered scenarios, even when the array aspect ratio is not very large.
The analytical PEB also closely matches the exact results over the majority range of $\gamma$, validating the effectiveness of the theoretical analysis.
The deviation between the analytical PEB and the exact results observed at $N = 512$ with a large $\gamma$ stems from the approximation of $(N_{\mathrm{y}}-1)d \approx N_{\mathrm{y}}d$ applied during the derivation.
In addition, as the distance increases or the number of antennas decreases, both the PEB and the optimal array aspect ratio $\gamma_{\mathrm{opt}}$ increase, indicating that $\gamma_{\mathrm{opt}}$ is jointly governed by the distance and the number of antennas.

\section{Codebook Design and Channel Estimation}\label{s:ChannelEst}
In channel estimation for XL-MIMO systems, compressive sensing techniques are widely applied to reduce pilot overhead \cite{SYue-24TWC}, \cite{HLei-24TSP}, \cite{YKang-25TCOM}.
Conventional codebook designs primarily rely on isotropic wavefront assumptions, employing the 2D-DFT codebook for far-field planar wavefronts \cite{YWang-21TWC} and the 3D polar-domain codebook for near-field spherical wavefronts \cite{CD-22TCOMM}.
However, these codebooks encounter limitations within the anisotropic near-field region of non-square UPAs.
Specifically, the 2D-DFT codebook fails to capture the quadratic phase along the long axis, causing energy leakage and beam spreading, while the 3D polar-domain codebook introduces over-parameterization and prohibitive computational overhead.
To address these challenges, we leverage the revealed anisotropic wavefront characteristics to design a novel 3D-ANF codebook, facilitating low-complexity channel estimation.
% To address these challenges, we design a novel 3D-ANF codebook by leveraging the revealed anisotropic wavefront characteristics, where the long and short axes simultaneously exhibit near-field spherical and far-field planar properties, respectively.

By exploiting the two-dimensional geometry of the non-square UPA, the codebook $\mathbf{D}_{\mathrm{ANF}}$ can be formulated as the Kronecker product of the long-axis dictionary $\mathbf{D}_{\mathrm{x}}$ and the short-axis dictionary $\mathbf{D}_{\mathrm{y}}$, i.e., $\mathbf{D}_{\mathrm{ANF}} = \mathbf{D}_{\mathrm{x}} \otimes \mathbf{D}_{\mathrm{y}}$.
To make $\mathbf{D}_{\mathrm{x}}$ well capture the spatial sparsity, we first analyze the equivalent noiseless received signal along the long axis $s(n_{\mathrm{x}})$, which can be modeled as a discrete chirp signal, as follows
\begin{align}
    s(n_{\mathrm{x}}) = \frac{1}{\sqrt{N_\mathrm{x}}} e^{ j2\pi \left( \tau_1 n_{\mathrm{x}} - \frac{\tau_2}{2} n_{\mathrm{x}}^2 \right) },
\end{align}
where $n_{\mathrm{x}} \in \mathcal{N}_{\mathrm{x}}$, $\tau_1 \triangleq \frac{u_{\mathrm{x}} d}{\lambda}$ and $\tau_2 \triangleq \frac{(1-{u_\mathrm{x}}^2)d^2}{\lambda r}$ denote the linear spatial frequency and the spatial chirp rate, respectively.
Due to the presence of the quadratic phase term, applying the DFT to $s(n_{\mathrm{x}})$ inevitably leads to spectrum spreading and energy leakage, thereby destroying the sparse representation of the signal in the angular domain.
In contrast, the DFrFT can rotate the spatial-frequency plane to provide a matched orthogonal basis for the chirp signal, thus restoring the channel sparsity in the fractional transform domain \cite{XYang-24WCL}.
The $p$-th order ($p \in [0.5, 1.5]$) DFrFT of $s(n_{\mathrm{x}})$ is derived as \eqref{eq:DFrFT_s}, shown at the top of the next page,
\begin{figure*}
\begin{align}\label{eq:DFrFT_s}
    S_\alpha(q) =& c_\alpha \sum_{n_{\mathrm{x}} = -\frac{N_{\mathrm{x}}-1}{2}}^{\frac{N_{\mathrm{x}}-1}{2}} e^{ j\frac{\pi \csc\alpha}{N_{\mathrm{x}}} \left( q^2 \cos\alpha - 2q n_{\mathrm{x}} + n_{\mathrm{x}}^2 \cos\alpha \right)} s(n_{\mathrm{x}}) \notag\\
    =& \frac{c_\alpha }{\sqrt{N_{\mathrm{x}}}} e^{j\pi \frac{q^2 \cot\alpha}{N_{\mathrm{x}}}} \sum_{n_{\mathrm{x}} = -\frac{N_{\mathrm{x}}-1}{2}}^{\frac{N_{\mathrm{x}}-1}{2}} e^{ j\pi n_{\mathrm{x}}^2 \left( \frac{\cot\alpha}{N_{\mathrm{x}}} - \tau_2 \right) - j2\pi n_{\mathrm{x}} \left( \frac{q \csc\alpha}{N_{\mathrm{x}}} - \tau_1 \right)}
\end{align}
{\noindent}\rule[-5pt]{17.5cm}{0.05em}
\end{figure*}
where $c_\alpha \triangleq \sqrt{\frac{1-j\cot\alpha}{N_{\mathrm{x}}}}$, $\alpha \triangleq \frac{p\pi}{2}$, and $q \in \mathcal{N}_{\mathrm{x}}$ represents the index of the DFrFT.
When the rotation angle of the DFrFT matches the spatial chirp rate of $s(n_{\mathrm{x}})$, i.e., $\cot \hat{\alpha} = N_{\mathrm{x}} \tau_2$, the quadratic phase term in \eqref{eq:DFrFT_s} is eliminated, simplifying the expression to the sum of a finite geometric progression, given by
\begin{align}
    S_{\hat{\alpha}}(q) = \frac{c_{\hat{\alpha}}}{\sqrt{N_{\mathrm{x}}}} e^{j\pi \frac{q^2 \cot\hat{\alpha}}{N_{\mathrm{x}}}} \sum_{n_{\mathrm{x}}=-\frac{N_{\mathrm{x}}-1}{2}}^{\frac{N_{\mathrm{x}}-1}{2}} e^{-j2\pi n_{\mathrm{x}} \left( \frac{q \csc\hat{\alpha}}{N_{\mathrm{x}}} - \tau_1 \right)}.
\end{align}
Then, $|S_{\hat{\alpha}}(q)|$ is calculated as
\begin{align}\label{eq:S_abs}
    |S_{\hat{\alpha}}(q)| = \frac{|c_{\hat{\alpha}}|}{\sqrt{N_{\mathrm{x}}}} \times \left| \frac{\sin\left( N_{\mathrm{x}} \pi \left( \frac{q \csc\hat{\alpha}}{N_{\mathrm{x}}} - \tau_1 \right) \right)}{\sin\left( \pi \left( \frac{q \csc\hat{\alpha}}{N_{\mathrm{x}}} - \tau_1 \right) \right)} \right|.
\end{align}
From \eqref{eq:S_abs}, it can be observed that the signal energy focuses into a Dirichlet kernel in the DFrFT domain, revealing the inherent sparsity of the near-field channel in the fractional transform domain.
To exploit this sparsity, we propose to design the dictionary $\mathbf{D}_{\mathrm{x}}$ by using the transformation kernel of the DFrFT.
Specifically, each column of $\mathbf{D}_{\mathrm{x}}$ is defined as a discrete chirp basis function $\mathbf{d}(q, s) \in \mathbb{C}^{N_{\mathrm{x}} \times 1}$, generated by jointly sampling the parameters $\tau_1$ and $\tau_2$, expressed as
\begin{align}
    [\mathbf{d}(q, s)]_i = \frac{1}{\sqrt{N_{\mathrm{x}}}} e^{ j2\pi \left( \tau_1(q) n_{\mathrm{x},i} - \frac{\tau_2(q, s)}{2} n_{\mathrm{x},i}^2 \right)},
\end{align}
where $\tau_1(q)$ and $\tau_2(q, s)$ represent the sampling grids for the linear spatial frequency and the spatial chirp rate, respectively.

Then, we determine the sampling spacings for $\tau_1$ and $\tau_2$ based on codeword correlations, and their sampling ranges follow from the admissible values of the long-axis direction cosine and the anisotropic near-field boundaries, respectively.
To set the sampling spacing in $\tau_1$, we compare codewords with the same chirp rate.
In this case, the correlation between codewords follows a Dirichlet kernel as a function of their linear spatial-frequency separation.
The first positive zero occurs at a separation of $\frac{1}{N_{\mathrm{x}}}$, which we adopt as the sampling spacing.
To set the sampling spacing in $\tau_2$, we evaluate the correlation between two long-axis codewords $\mathbf{d}(q, s_1)$ and $\mathbf{d}(q, s_2)$ that are adjacent in chirp rate at a fixed linear spatial frequency.
Approximating the discrete summation by an integral, we obtain
\begin{align}
    \Xi(\Delta \tau) =& \left| \mathbf{d}^H(q, s_1)\mathbf{d}(q, s_2) \right| \notag\\
    =& \left| \frac{1}{N_{\mathrm{x}}} \sum_{n_{\mathrm{x}} = -\frac{N_{\mathrm{x}}-1}{2}}^{\frac{N_{\mathrm{x}}-1}{2}} e^{-j\pi \Delta \tau n_{\mathrm{x}}^2} \right| \notag\\
    \approx & \left| \frac{1}{N_{\mathrm{x}}} \int_{-\frac{N_{\mathrm{x}}}{2}}^{\frac{N_{\mathrm{x}}}{2}} e^{-j\pi\Delta\tau x^2}\,\mathrm{d}x \right| \notag\\
    =& \frac{\sqrt{C^2\left(\frac{N_{\mathrm{x}}}{2} \sqrt{2 \Delta \tau}\right) + S^2\left(\frac{N_{\mathrm{x}}}{2} \sqrt{2 \Delta \tau}\right)}}{\frac{N_{\mathrm{x}}}{2} \sqrt{2 \Delta \tau}},
\end{align}
where $\Delta \tau \triangleq |\tau_{2,1} - \tau_{2,2}|$ denotes the sampling spacing.
Based on the properties of Fresnel integrals, $\Xi(\Delta \tau)$ exhibits an oscillatory decay as $\Delta \tau$ increases.
Thus, the sampling spacing should be aligned with the first local minimum of $\Xi(\Delta \tau)$ to achieve quasi-orthogonality between adjacent codewords.
By setting $\frac{d \Xi(\Delta \tau)}{d\Delta \tau} = 0$, the sampling spacing is determined to be $\Delta \tau \approx \frac{7}{N_{\mathrm{x}}^2}$.
The sampling range of $\tau_1$ follows from the admissible range of the long-axis direction cosine $u_{\mathrm{x}}$.
At half-wavelength antenna spacing, $\tau_1=\frac{u_{\mathrm{x}}}{2}$, so $u_{\mathrm{x}} \in [-1, 1]$ maps to
$\tau_1 \in [-\frac{1}{2}, \frac{1}{2}]$.
Using the spacing derived above, i.e., $\frac{1}{N_{\mathrm{x}}}$, we set $\tau_1(q) = \frac{q}{N_{\mathrm{x}}}$ for $q \in \mathcal{N}_{\mathrm{x}}$, which gives $u_{\mathrm{x}}(q)=\frac{2q}{N_{\mathrm{x}}}$.
For $\tau_2$, we determine the sampling range for each $q$ from the anisotropic near-field boundaries.
At a given observation direction $(u_{\mathrm{x}}(q), u_{\mathrm{y}})$, the chirp rate is proportional to
$\frac{1}{r}$, so the distance interval $R_{\mathrm{y}} \leq r < R_{\mathrm{x}}$ gives
\begin{align}
    \frac{2\eta_0^2}{N_{\mathrm{x}}^2} < \tau_2 \leq 2\eta_0^2\frac{1 - u_{\mathrm{x}}^2(q)}{N_{\mathrm{y}}^2(1 - u_{\mathrm{y}}^2)}.
\end{align}
For a prescribed fixed elevation sector $|\theta| \leq \theta_{\max} < \frac{\pi}{2}$, we adopt common sampling bounds as follows
\begin{subequations}
    \begin{align}
        \tau_{2,\min}^{(q)} =& \frac{2\eta_0^2}{N_{\mathrm{x}}^2}, \\
        \tau_{2,\max}^{(q)} =& 2\eta_0^2\frac{1 - (2q/N_{\mathrm{x}})^2 }{N_{\mathrm{y}}^2 \cos^2 \theta_{\max}}.
    \end{align}
\end{subequations}
The resulting sampling range for each $q$ covers the chirp-rate intervals for all observation directions associated with that $q$ within the sector.
Thus, we have $\tau_2(q, s) = \tau_{2,\min}^{(q)} + s \Delta \tau$ for $s \in \{0, 1, \dots, S_q-1\}$, where $S_q \triangleq \lceil \frac{\tau_{2,\max}^{(q)} - \tau_{2,\min}^{(q)}}{\Delta \tau} \rceil$ is the number of sampling grids for the spatial chirp rate and $\Delta \tau = \frac{7}{N_{\mathrm{x}}^2}$ as we determined before.
Moreover, since the user is located in the anisotropic near-field region, the short-axis response contains only a linear phase.
Therefore, $\mathbf{D}_{\mathrm{y}} \in \mathbb{C}^{N_{\mathrm{y}} \times N_{\mathrm{y}}}$ adopts the standard DFT dictionary with spatial-frequency spacing $\frac{1}{N_{\mathrm{y}}}$.

\begin{algorithm}[!htbp]
    \caption{ANF-OMP Algorithm}\label{alg:ANF_OMP}
    \begin{algorithmic}[1]
        \STATE \textbf{Initialize:} $\mathbf{y}_{\mathrm{r}}^{(0)} = \mathbf{y}_2$ and $\Lambda_0 = \emptyset$.
        \FOR{$t = 1, \ldots, P$}
            \STATE Obtain $\mathbf{E}_{t-1} = \mathrm{reshape}(\mathbf{y}_{\mathrm{r}}^{(t-1)},N_{\mathrm{y}},N_{\mathrm{x}})$.
            \STATE Obtain $j_t = \arg\max_{j \notin \Lambda_{t-1}} |[\mathrm{vec}(\mathbf{D}_{\mathrm{y}}^H\mathbf{E}_{t-1}\mathbf{D}_{\mathrm{x}}^*)]_j|$.
            \STATE Update $\Lambda_t = \Lambda_{t-1} \cup \{j_t\}$.
            \STATE Obtain $\hat{\mathbf{g}}_t = \arg\min_{\mathbf{g}} \|\mathbf{y}_2 - \mathbf{D}_{\Lambda_t}\mathbf{g}\|^2$.
            \STATE Update $\mathbf{y}_{\mathrm{r}}^{(t)} = \mathbf{y}_2 - \mathbf{D}_{\Lambda_t}\hat{\mathbf{g}}_t$.
        \ENDFOR
        \STATE Recover $\hat{\mathbf{h}} = \mathbf{D}_{\Lambda_P}\hat{\mathbf{g}}_P$.
    \end{algorithmic}
\end{algorithm}

\begin{table*}[!htbp]
       \footnotesize
       \centering
       \caption{Codebook Size and Computational Complexity}\label{table:Complexity}
       \begin{tabular}{|c|c|c|}
       \hline
       Codebook (Algorithm)  & Codebook Size   & Computational Complexity   \\ \cline{1-3}
       3D-ANF (ANF-OMP)     & $ \nu^3N \lceil \frac{2\eta_0^2}{7}(\frac{\gamma^2}{\cos^2 \theta_{\max}}-1) \rceil $ & $\mathcal{O}\{ P\nu N \log(\nu N_{\mathrm{y}}) + P\nu^3 \gamma^2 N \log(\nu N_{\mathrm{x}})\}$ \\ \cline{1-3}
       3D-Polar (P-OMP) \cite{CD-22TCOMM} & $\nu^3N S_{\mathrm{r}}$     & $\mathcal{O}( P\nu^3N^2 S_{\mathrm{r}})$ \\ \cline{1-3}
       2D-DFT (FF-OMP) \cite{YWang-21TWC}  & $\nu^2N$  & $\mathcal{O}\{P\nu^2 N \log(\nu^2N)\}$   \\ \hline
       \end{tabular}
\end{table*}

Based on the proposed codebook, the OMP-based algorithm can be employed for near-field channel estimation.
Specifically, the ANF-OMP algorithm is summarized in {\bf Algorithm~\ref{alg:ANF_OMP}}.
In this algorithm, $\mathbf{y}_2 \in \mathbb{C}^{N \times 1}$ denotes the received pilot observation for multipath channel estimation, and $M$ is the number of codewords.
At iteration $t$, $\mathbf{y}_{\mathrm{r}}^{(t)} \in \mathbb{C}^{N \times 1}$ denotes the residual, and $\Lambda_t$ contains the indices of the selected codewords.
The subdictionary $\mathbf{D}_{\Lambda_t}$ consists of the columns of $\mathbf{D}_{\mathrm{ANF}}$ indexed by $\Lambda_t$, and $\hat{\mathbf{g}}_t \in \mathbb{C}^{t \times 1}$ contains the corresponding least-squares coefficient estimates.
The output $\hat{\mathbf{h}} \in \mathbb{C}^{N \times 1}$ is the estimated multipath channel vector.
Note that {\bf Algorithm~\ref{alg:ANF_OMP}} exploits the Kronecker structure to express codeword matching as $\mathrm{vec}(\mathbf{D}_{\mathrm{y}}^H\mathbf{E}_{t-1}\mathbf{D}_{\mathrm{x}}^*)$, where the left multiplication is performed via short-axis FFTs and the right multiplication via quadratic-phase compensation followed by long-axis FFTs, thereby reducing the computational complexity.

\begin{figure}[htbp]
       \centering
       \includegraphics[width=0.43\textwidth]{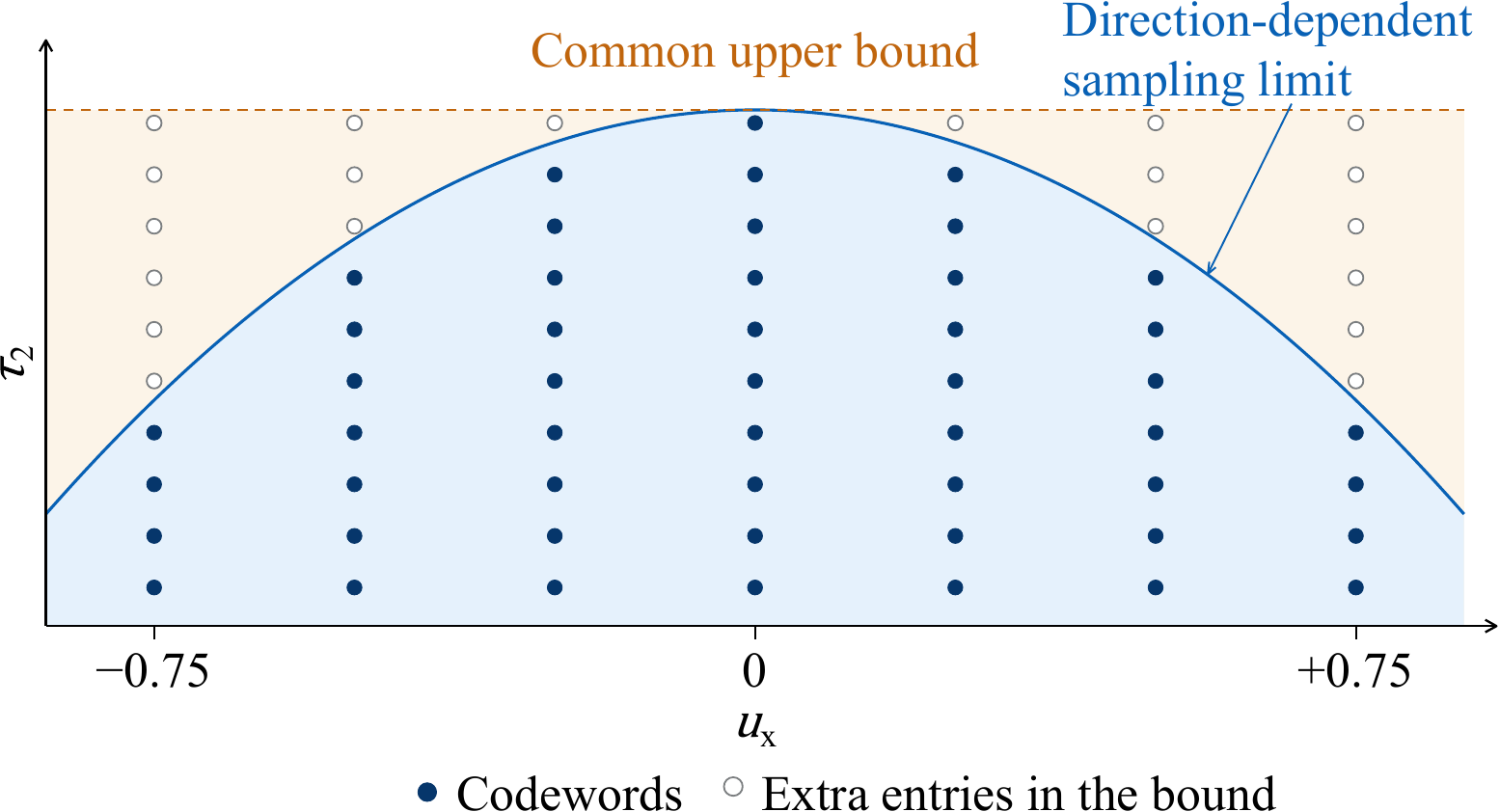}
       \caption{Long-axis codebook sampling and the upper bound on codebook size.}\label{fig:codebook_size}
\end{figure}

Table\,\ref{table:Complexity} summarizes the codebook size and computational complexity comparison of the ANF-OMP, P-OMP \cite{CD-22TCOMM}, and FF-OMP \cite{YWang-21TWC} algorithms, which utilize the proposed 3D-ANF, 3D-polar, and 2D-DFT codebooks, respectively.
$P$ is the number of channel paths, $\nu$ is the oversampling factor, and $S_{\mathrm{r}}$ is the number of distance-domain sampling grids.
Note that the 3D-ANF codebook size reported in Table\,\ref{table:Complexity} is an upper bound used to assess the worst-case computational requirements, as illustrated in Fig.\,\ref{fig:codebook_size}.
It is observed that the computational complexity of the P-OMP algorithm exhibits a quadratic growth with respect to the number of antennas.
By contrast, the proposed ANF-OMP algorithm achieves a significant complexity reduction by exploiting the Kronecker product structure of the anisotropic codebook to decouple the high-dimensional codeword matching into serial fast Fourier transform (FFT) operations, thereby meeting the low-overhead requirements of practical XL-MIMO systems.

\begin{figure}[htbp]
       \centering
       \subfloat[]{
       \includegraphics[width=0.235\textwidth]{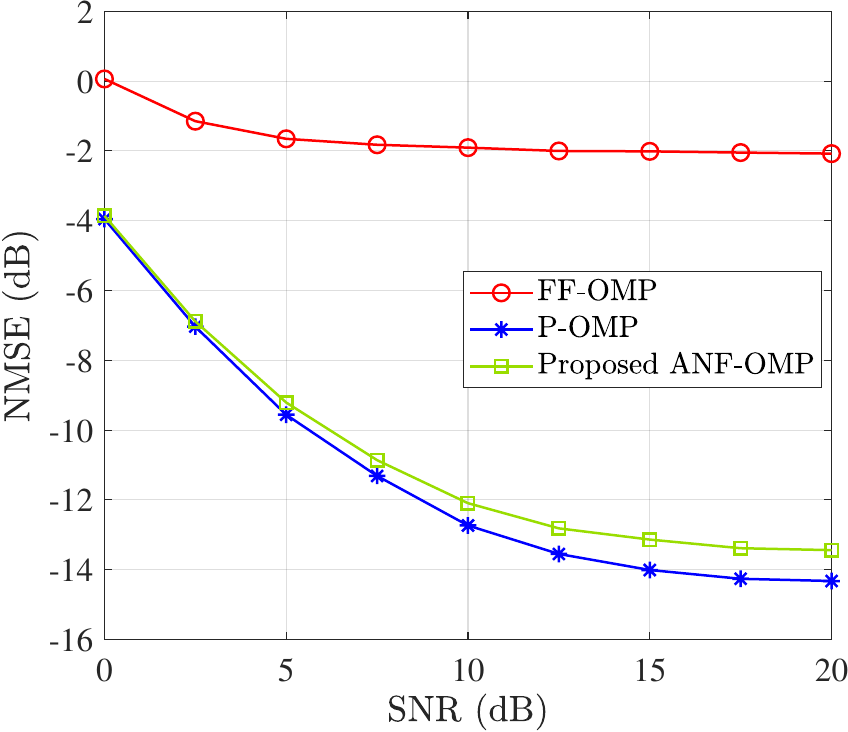}}
       \subfloat[]{
       \includegraphics[width=0.235\textwidth]{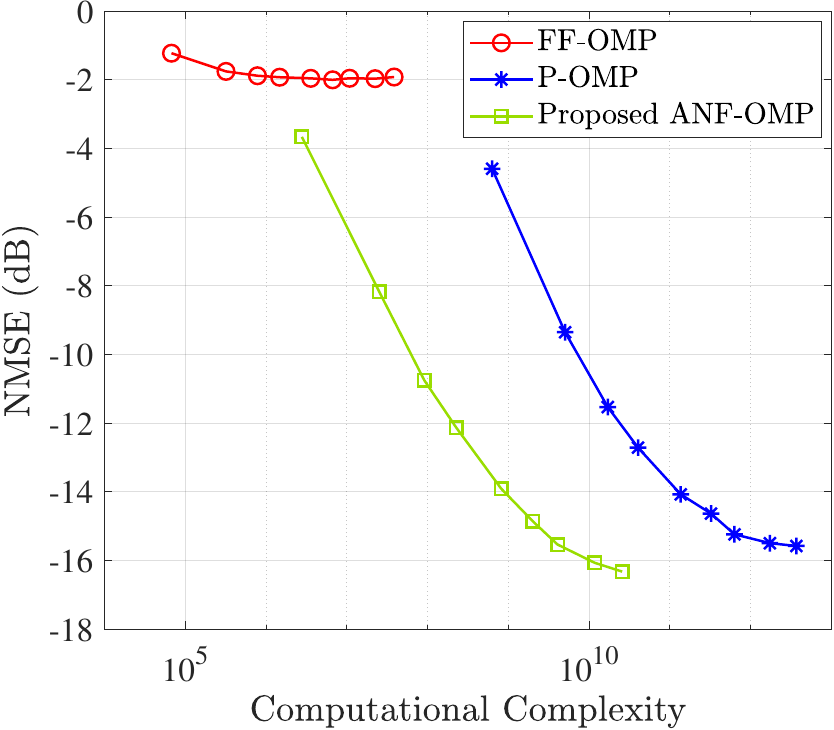}}\\
       \caption{(a) The NMSE of channel estimation versus the SNR for different algorithms, (b) the NMSE of channel estimation versus the computational complexity for different algorithms.}\label{fig:NMSEh}
\end{figure}

\begin{table}[htbp]
       \footnotesize
       \centering
       \caption{Average Run Time}\label{table:RunTime}
       \begin{tabular}{|c|c|}
       \hline
       Algorithm  & Average Run Time   \\ \cline{1-2}
       ANF-OMP   & $0.19\,\mathrm{s}$ \\ \cline{1-2}
       P-OMP  \cite{CD-22TCOMM}   & $5.56\,\mathrm{s}$ \\ \cline{1-2}
       FF-OMP \cite{YWang-21TWC}   & $0.02\,\mathrm{s}$   \\ \hline
       \end{tabular}
\end{table}

The normalized mean squared error (NMSE) of channel estimation versus the SNR for different algorithms is illustrated in Fig.\,\ref{fig:NMSEh}(a), where $N_{\mathrm{x}} = 128$, $N_{\mathrm{y}} = 16$, $P = 3$, $\nu = 4$, and the carrier frequency is $28\,\mathrm{GHz}$.
For each path, the distance is uniformly distributed over $[1.38, 10.10]\,\mathrm{m}$, while the elevation and azimuth AoAs are independently and uniformly distributed over $[-30^\circ,30^\circ]$.
We perform $1000$ channel realizations at each SNR for the NMSE, defined as $\mathrm{NMSE} = \mathbb{E}(\|\mathbf{h}-\hat{\mathbf{h}}\|^2/\|\mathbf{h}\|^2)$, where $\mathbf{h}$ and $\hat{\mathbf{h}}$ denote the actual and estimated channel vectors, respectively.
Constrained by the planar wavefront assumption, the FF-OMP algorithm fails to capture the quadratic phase within the anisotropic near-field region, leading to NMSE degradation across the considered SNR regime.
Conversely, the P-OMP algorithm achieves superior NMSE performance through joint sampling in the angular and distance domains.
Notably, the proposed ANF-OMP algorithm attains comparable NMSE performance to the P-OMP algorithm, validating its effectiveness.
To evaluate the feasibility in practical deployments, Table\,\ref{table:RunTime} compares the average run time of the considered algorithms.
All three algorithms are implemented in MATLAB R2023a with GPU acceleration on a server equipped with an Intel Core i9-13900KF processor and an NVIDIA GeForce RTX 4090 GPU.
The FF-OMP algorithm requires the lowest computational overhead, and the P-OMP algorithm incurs a prohibitive computational cost.
In contrast, the proposed ANF-OMP algorithm decouples the 3D search space into separate long- and short-axis operations and utilizes the FFT to accelerate the estimation, thereby achieving an approximately $96.6\%$ reduction in average run time compared to the P-OMP algorithm without compromising estimation accuracy.

Fig.\,\ref{fig:NMSEh}(b) depicts the NMSE of channel estimation versus the computational complexity for different algorithms under the same system configuration as Fig.\,\ref{fig:NMSEh}(a), where the operating points are generated by varying the codebook oversampling factor $\nu$ over $\{1,2,3,4,6,8\}$.
As can be shown, the FF-OMP algorithm suffers from a severe error floor due to physical model mismatches, thereby failing to improve estimation accuracy even with additional computational resources.
Although the P-OMP algorithm can achieve high estimation accuracy, its excessive computational overhead limits its application in practical XL-MIMO systems.
Compared to the P-OMP algorithm, the proposed ANF-OMP algorithm maintains comparable NMSE performance with substantially lower computational complexity, establishing its practical advantage.

\section{Conclusion}\label{s:Conclusion}
In this paper, we investigated the anisotropic near-field characteristics, fundamental limits, and channel estimation for non-square UPA-enabled XL-MIMO systems.
Based on the derived EBRDs for both the long and short axes of the array, the radiation space was partitioned into the fully near-field, the anisotropic near-field, and the far-field regions.
The analyses revealed that the anisotropic region dominates as the array aspect ratio increases.
An asymptotic closed-form EDoF expression was also obtained, demonstrating that the distance-domain spatial multiplexing capability is governed by the long-axis aperture in the large array aspect ratio regime.
Furthermore, we derived the closed-form distance estimation CRB and the 3D PEB, thereby identifying the optimal array aspect ratio that minimizes the 3D positioning error.
Finally, by exploiting the anisotropic wavefront properties, we proposed a 3D-ANF codebook and a low-complexity ANF-OMP algorithm.
Numerical results verified that the proposed algorithm achieves estimation accuracy comparable to the P-OMP algorithm with significantly reduced computational overhead.

\appendices
\section{Proof of Theorem \ref{the:Ry}}\label{proof:the:Ry}
Let the constant $\eta_0$ denote the primary root of the half-power equation $\frac{C^2(\eta_{\mathrm{y}}) + S^2(\eta_{\mathrm{y}})}{\eta_{\mathrm{y}}^2} = 0.5$.
According to the definition of the $3\,\mathrm{dB}$ beam depth, the effective beamfocusing range of the short axis satisfies $g_{\mathrm{y}}(r_{\mathrm{F}}, u_{\mathrm{y}}) \geq 0.5$.
% Since $g_{\mathrm{y}}(r_{\mathrm{F}}, u_{\mathrm{y}})$ decreases monotonically with respect to $\eta_{\mathrm{y}}$, this inequality is equivalent to
Considering the main-lobe region and the primary half-power crossing of $\frac{C^2(\eta_{\mathrm{y}}) + S^2(\eta_{\mathrm{y}})}{\eta_{\mathrm{y}}^2}$, the condition $g_{\mathrm{y}}(r_{\mathrm{F}}, u_{\mathrm{y}}) \geq 0.5$ is equivalent to
\begin{align}
    \frac{N_{\mathrm{y}}^2 d^2 (1-u_{\mathrm{y}}^2)}{2\lambda}\left|\frac{1}{r_{\mathrm{F}}} - \frac{1}{r}\right| \leq \eta_0^2.
\end{align}
Consequently, the beamfocusing range of the short axis in the distance domain is given by
\begin{align}
    -\frac{2\lambda \eta_0^2}{N_{\mathrm{y}}^2 d^2 (1-u_{\mathrm{y}}^2)} \leq \frac{1}{r_{\mathrm{F}}} - \frac{1}{r} \leq \frac{2\lambda \eta_0^2}{N_{\mathrm{y}}^2 d^2 (1-u_{\mathrm{y}}^2)}.
\end{align}
Rearranging this inequality yields $r_{\mathrm{lower}} \leq r \leq r_{\mathrm{upper}}$, where $r_{\mathrm{lower}}$ and $r_{\mathrm{upper}}$ denote the lower and upper bounds of the beamfocusing range in the distance domain, respectively, which are expressed as
\begin{subequations}
    \begin{align}
        r_{\mathrm{lower}} =& \frac{r_{\mathrm{F}} N_{\mathrm{y}}^2 d^2 (1-u_{\mathrm{y}}^2)}{N_{\mathrm{y}}^2 d^2 (1-u_{\mathrm{y}}^2) + 2\lambda \eta_0^2 r_{\mathrm{F}}}, \\
        r_{\mathrm{upper}} =& \frac{r_{\mathrm{F}} N_{\mathrm{y}}^2 d^2 (1-u_{\mathrm{y}}^2)}{N_{\mathrm{y}}^2 d^2 (1-u_{\mathrm{y}}^2) - 2\lambda \eta_0^2 r_{\mathrm{F}}}.
    \end{align}
\end{subequations}
Accordingly, the $3\,\mathrm{dB}$ beam depth of the short axis $r_{\mathrm{BD,y}}$ is derived as
\begin{align}\label{eq:r_BDy}
    r_{\mathrm{BD,y}} = \frac{4\lambda \eta_0^2 r_{\mathrm{F}}^2 N_{\mathrm{y}}^2 d^2 (1-u_{\mathrm{y}}^2)}{\left[ N_{\mathrm{y}}^2 d^2 (1-u_{\mathrm{y}}^2) \right]^2 - 4\lambda^2 \eta_0^4 r_{\mathrm{F}}^2}.
\end{align}

As the focal distance $r_{\mathrm{F}}$ increases and reaches a threshold, the short axis loses its beamfocusing capability, yielding an infinite beam depth, i.e., $r_{\mathrm{BD,y}} \rightarrow \infty$.
Mathematically, this condition is equivalent to setting the denominator of \eqref{eq:r_BDy} to zero.
Solving this equality yields the EBRD $R_{\mathrm{y}}$ as presented in \eqref{eq:Ry}, which completes the proof.

\section{Proof of Theorem \ref{the:K}}\label{proof:the:K}
The EBRD of the entire array, i.e., $R_{\mathrm{array}}$, is defined such that when the focal distance is $r_{\mathrm{F}} \geq R_{\mathrm{array}}$, the $3\,\mathrm{dB}$ beam depth of the array is infinite.
That is, at $r_{\mathrm{F}} = R_{\mathrm{array}}$, the normalized array gain tends to $0.5$ as the observation distance approaches infinity, yielding
\begin{align}\label{eq:g-R_array}
    g_{\mathrm{x}}(R_{\mathrm{array}},0) \times g_{\mathrm{y}}(R_{\mathrm{array}},0) = 0.5.
\end{align}
In this case, $z_{\mathrm{eff}} = \frac{1}{R_{\mathrm{array}}}$, $\eta_{\mathrm{x}} = \sqrt{\frac{N_{\mathrm{x}}^2 d^2}{2\lambda R_{\mathrm{array}}}}$, and $\eta_{\mathrm{y}} = \sqrt{\frac{N_{\mathrm{y}}^2 d^2}{2\lambda R_{\mathrm{array}}}}$, which implies $\eta_{\mathrm{x}} = \gamma \eta_{\mathrm{y}}$. 
Based on \eqref{eq:gain_Fresnel_x} and introducing the gain function $F(v) = \frac{C^2(v) + S^2(v)}{v^2}$ defined by the Fresnel integrals, solving \eqref{eq:g-R_array} is equivalent to solving the following transcendental equation
\begin{align}\label{eq:F-eta_x-equation}
    F(\eta_{\mathrm{x}}) \times F\left(\frac{\eta_{\mathrm{x}}}{\gamma}\right) = 0.5.
\end{align}
Similarly, $R_{\mathrm{x}}$ is the root of the equation $F(\eta_0) = 0.5$, where $\eta_0 = \sqrt{\frac{N_{\mathrm{x}}^2 d^2}{2\lambda R_{\mathrm{x}}}}$.
Applying the Maclaurin expansion to the gain function $F(v)$ yields
\begin{align}\label{eq:Fv-Maclaurin}
    F(v) \approx & \frac{\left\{\int_0^v \left(1 - \frac{\pi^2}{8}t^4 \right) dt\right\}^2 + \left\{\int_0^v \frac{\pi}{2}t^2 dt\right\}^2}{v^2} \notag\\
    =& 1 - \frac{\pi^2}{45}v^4.
\end{align}
Substituting $v = \frac{\eta_{\mathrm{x}}}{\gamma}$ into \eqref{eq:Fv-Maclaurin}, we obtain
\begin{align}\label{eq:F-eta_x-gamma}
    F\left(\frac{\eta_{\mathrm{x}}}{\gamma}\right) = 1 - \frac{\pi^2}{45\gamma^4}\eta_{\mathrm{x}}^4.
\end{align}
Substituting \eqref{eq:F-eta_x-gamma} into \eqref{eq:F-eta_x-equation} and further applying the Maclaurin expansion leads to $F(\eta_{\mathrm{x}}) \approx 0.5 + \frac{\pi^2}{90 \gamma^4} \eta_{\mathrm{x}}^4$.

Let $\eta_{\mathrm{x}} = \eta_0 + \Delta \eta$, where $\Delta \eta$ denotes the deviation of $\eta_{\mathrm{x}}$ from $\eta_0$. Thus, we have
\begin{subequations}
\begin{align}
    F(\eta_{\mathrm{x}}) \overset{\mathrm{(c)}}{\approx}& 0.5 + F'(\eta_0) \Delta \eta, \\
    F(\eta_0 + \Delta \eta) \approx& 0.5 + \frac{\pi^2}{90\gamma^4}(\eta_0 + \Delta \eta)^4,
\end{align}
\end{subequations}
where the approximation $\mathrm{(c)}$ follows from the first-order Taylor expansion, and $F'(v)$ is the first derivative of $F(v)$.
Based on the properties of the Fresnel integrals, $F'(\eta_0)$ is strictly negative. 
Since $F(\eta_{\mathrm{x}}) = F(\eta_0 + \Delta \eta)$, we obtain $\Delta \eta \approx \frac{\pi^2\eta_0^4}{90F'(\eta_0)\gamma^4}$.
Consequently, $K(\gamma)$ can be formulated as
\begin{align}
    K(\gamma) &= 1 - \frac{R_{\mathrm{x}}}{R_{\mathrm{array}}} \notag\\
    &= 1 - \frac{\eta_{\mathrm{x}}^2}{\eta_0^2} \notag\\
    &\overset{\mathrm{(d)}}{\approx} -\frac{2\Delta \eta}{\eta_0} \notag\\
    &\approx \frac{\pi^2\eta_0^3}{45 |F'(\eta_0)| \gamma^4},
\end{align}
where the approximation $\mathrm{(d)}$ is derived from $\eta_{\mathrm{x}} = \eta_0 + \Delta \eta$ and the first-order Taylor expansion.
Furthermore, $\bar{K}(\gamma)$ can be expressed as
\begin{align}
    \bar{K}(\gamma) =& \left(1 - \frac{1}{\gamma^2}\right) \frac{R_{\mathrm{x}}}{R_{\mathrm{array}}} \notag\\
    =& \left(1 - \frac{1}{\gamma^2}\right) \times \left(1 - \frac{\varsigma}{\gamma^4}\right) \notag\\
    =& 1 - \frac{1}{\gamma^2} - \mathcal{O}\left(\frac{1}{\gamma^4}\right),
\end{align}
where $\varsigma \triangleq \frac{\pi^2\eta_0^3}{45 |F'(\eta_0)|}$. This completes the proof.

\section{Proof of Theorem \ref{the:EDoF}}\label{proof:the:EDoF}
Along the broadside direction, define $W_{\zeta} \triangleq \zeta_{\mathrm{y}} - \zeta_{\mathrm{x}} > 0$ and $\xi \triangleq \frac{N \gamma d^2(\zeta_{\mathrm{y}}-\zeta_{\mathrm{x}})}{2\lambda} \gg 1$.
Conditioned on the user lying in the anisotropic near-field region, $\zeta$ is uniformly distributed over $(\zeta_{\mathrm{x}}, \zeta_{\mathrm{y}})$, and the conditional covariance satisfies
\begin{align}
    \mathbf{R}_{\mathrm{A}} =& \frac{1}{W_{\zeta}} \int_{\zeta_{\mathrm{x}}}^{\zeta_{\mathrm{y}}} \mathbf{a}_{\mathrm{ANF}}(1/\zeta,0,0)\mathbf{a}^H_{\mathrm{ANF}}(1/\zeta,0,0) d\zeta \notag\\
    =& \frac{1}{N_{\mathrm{y}}} \mathbf{R}_{\mathrm{x}} \otimes \mathbf{1}_{N_{\mathrm{y}} \times N_{\mathrm{y}}},
\end{align}
where $\mathbf{R}_{\mathrm{x}} \triangleq \frac{1}{W_{\zeta}} \int_{\zeta_{\mathrm{x}}}^{\zeta_{\mathrm{y}}} \mathbf{a}_{\mathrm{x,NF}}(1/\zeta, 0)\mathbf{a}^H_{\mathrm{x,NF}}(1/\zeta, 0) d\zeta$.
Thus, $\mathrm{tr}(\mathbf{R}_{\mathrm{A}}) = 1$ and $\mathrm{tr}(\mathbf{R}_{\mathrm{A}}^2) = \mathrm{tr}(\mathbf{R}_x^2)$.
Evaluating the integral over $\zeta$ and employing the continuous-aperture approximation, we obtain
\begin{align}% 指数函数积分，取模
    \mathrm{tr}(\mathbf{R}_{\mathrm{A}}^2) =& \frac{1}{N_{\mathrm{x}}^2} \sum_{n= -\frac{N_{\mathrm{x}}-1}{2}}^{\frac{N_{\mathrm{x}}-1}{2}} \sum_{m= -\frac{N_{\mathrm{x}}-1}{2}}^{\frac{N_{\mathrm{x}}-1}{2}}
    \mathrm{sinc}^2 \left\{\frac{\xi(n^2-m^2)}{N_{\mathrm{x}}^2}\right\} \notag\\
    \approx & \int_{-\frac{1}{2}}^{\frac{1}{2}} \int_{-\frac{1}{2}}^{\frac{1}{2}} \mathrm{sinc}^2 \{ \xi(s^2 - t^2)\} d s d t,
\end{align}
where $\mathrm{sinc}(z)=\frac{\sin(\pi z)}{\pi z}$.
For $\xi \gg 1$, applying $\tilde{s} = s+t$, $\tilde{t} = s-t$, and subsequently $\chi = \xi\tilde{s}\tilde{t}$ yields
\begin{align}\label{eq:RA_tr}
    \mathrm{tr}(\mathbf{R}_{\mathrm{A}}^2) \approx & \frac{2}{\xi} \int_{0}^{\frac{\xi}{4}} \mathrm{sinc}^2(\chi)
    \ln\left( \frac{1+\sqrt{1-4\chi/\xi}}{1-\sqrt{1-4\chi/\xi}} \right) d\chi \notag\\
    \approx & \frac{2}{\xi} \int_{0}^{\frac{\xi}{4}} \mathrm{sinc}^2(\chi) \ln \frac{\xi}{\chi} d\chi \notag\\
    \approx & \frac{2\ln \xi}{\xi} \int_{0}^{\infty} \mathrm{sinc}^2(\chi) d\chi - \frac{2}{\xi} \int_{0}^{\infty} \mathrm{sinc}^2(\chi) \ln \chi d\chi.
\end{align}
Substituting $\int_{0}^{\infty} \mathrm{sinc}^2(\chi) d\chi = \frac{1}{2}$ and $\int_{0}^{\infty} \mathrm{sinc}^2(\chi)\ln\chi d\chi = \frac{1-\gamma_{\mathrm{E}}-\ln(2\pi)}{2}$ into \eqref{eq:RA_tr}, we obtain
\begin{align}\label{eq:R2_tr_ANF}
    \mathrm{tr}(\mathbf{R}_{\mathrm{A}}^2) \approx \frac{\ln \xi + \gamma_{\mathrm{E}} + \ln(2\pi) - 1}{\xi},
\end{align}
where $\gamma_{\mathrm{E}}$ denotes the Euler-Mascheroni constant.
Substituting \eqref{eq:R2_tr_ANF} into \eqref{eq:EDoF0} yields the asymptotic closed-form expression for the EDoF in \eqref{eq:EDoF}, which completes the proof.

\section{Proof of Theorem \ref{the:CRB_r}}\label{proof:the:CRB_r}
Given the statistical properties of the AWGN, the received signal $\mathbf{y}_1$ follows a complex Gaussian distribution, i.e., $\mathbf{y}_1 \sim \mathcal{CN}(\boldsymbol{\mu}(\boldsymbol{\vartheta}), \sigma_1^2 \mathbf{I})$, where $\boldsymbol{\mu}(\boldsymbol{\vartheta}) \triangleq \alpha \sqrt{N} \mathbf{a}_{\mathrm{NF}}(r, u_{\mathrm{x}}, u_{\mathrm{y}})$.
The corresponding log-likelihood function is given by
\begin{align}
    \ln p(\mathbf{y}_1 | \boldsymbol{\vartheta}) = -N \ln(\pi \sigma_1^2) -\frac{1}{\sigma_1^2} \left\| \mathbf{y}_1 - \boldsymbol{\mu}(\boldsymbol{\vartheta}) \right\|^2,
\end{align}
where $\boldsymbol{\vartheta} \triangleq [r, u_{\mathrm{x}}, u_{\mathrm{y}}, \varphi_{\alpha}]^T$ denotes the unknown parameter vector. 
According to the Slepian-Bangs formula \cite{Slepian-54TIT}, the Fisher information matrix (FIM) $\mathbf{J} \in \mathbb{R}^{4 \times 4}$ is expressed as
\begin{align}
    \mathbf{J}_{(i, j)} =& -\mathbb{E} \left\{ \frac{\partial^2 \ln p(\mathbf{y}_1 | \boldsymbol{\vartheta})}{\partial [\boldsymbol{\vartheta}]_i \partial [\boldsymbol{\vartheta}]_j} \right\} \notag\\
    =& \frac{2}{\sigma_1^2} \Re \left\{ \frac{\partial \boldsymbol{\mu}^H(\boldsymbol{\vartheta})}{\partial [\boldsymbol{\vartheta}]_i} \frac{\partial \boldsymbol{\mu}(\boldsymbol{\vartheta})}{\partial [\boldsymbol{\vartheta}]_j} \right\} \notag\\
    \overset{\mathrm{(e)}}{=}& \frac{2}{\sigma_1^2} \Re \left\{ \sum_{n=1}^{N} \left( j \frac{\partial [\boldsymbol{\psi}]_n}{\partial [\boldsymbol{\vartheta}]_i} [\boldsymbol{\mu}(\boldsymbol{\vartheta})]_n \right)^H \! \left( j \frac{\partial [\boldsymbol{\psi}]_n}{\partial [\boldsymbol{\vartheta}]_j} [\boldsymbol{\mu}(\boldsymbol{\vartheta})]_n \right) \! \right\} \notag\\
    \overset{\mathrm{(f)}}{=}& 2\rho_1 \sum_{n=1}^{N} \frac{\partial [\boldsymbol{\psi}]_n}{\partial [\boldsymbol{\vartheta}]_i} \frac{\partial [\boldsymbol{\psi}]_n}{\partial [\boldsymbol{\vartheta}]_j},
\end{align}
where $\boldsymbol{\psi} \triangleq \varphi_{r}(n) + \varphi_{\alpha}$ denotes the phase of the received signal with $\varphi_{r}(n) \triangleq -\frac{2\pi}{\lambda}r_{(n_{\mathrm{x}}, n_{\mathrm{y}})}$, and $\rho_1 = \frac{|\alpha|^2}{\sigma_1^2}$ is the received SNR.
The procedure $\mathrm{(e)}$ is obtained by substituting $[\boldsymbol{\mu}(\boldsymbol{\vartheta})]_n = |\alpha| e^{j [\boldsymbol{\psi}]_n}$, and the procedure $\mathrm{(f)}$ is derived based on $(j [\boldsymbol{\mu}(\boldsymbol{\vartheta})]_n)^H = -j [\boldsymbol{\mu}(\boldsymbol{\vartheta})]_n^*$ and $|[\boldsymbol{\mu}(\boldsymbol{\vartheta})]_n|^2 = |\alpha|^2$.
Subsequently, applying the Schur complement yields the equivalent FIM submatrix $\tilde{\mathbf{J}} \in \mathbb{R}^{3 \times 3}$, which contains only the spatial parameters $\tilde{\boldsymbol{\vartheta}} \triangleq [r, u_{\mathrm{x}}, u_{\mathrm{y}}]^T$, formulated as
\begin{align}
    \tilde{\mathbf{J}}_{(i, j)} =& \mathbf{J}_{(i, j)} - \frac{\mathbf{J}_{(i, \varphi_{\alpha})} \mathbf{J}_{(j, \varphi_{\alpha})}}{\mathbf{J}_{(\varphi_{\alpha},\varphi_{\alpha})}} \notag\\
    =& 2N\rho_1 \left\{ \frac{1}{N}\sum_{n=1}^{N} \frac{\partial \varphi_{r}(n)}{\partial [\boldsymbol{\vartheta}]_i} \frac{\partial \varphi_{r}(n)}{\partial [\boldsymbol{\vartheta}]_j} \right.  \notag\\
    & \hspace{3em} \left. - \left(\frac{1}{N}\sum_{n=1}^{N} \frac{\partial \varphi_{r}(n)}{\partial [\boldsymbol{\vartheta}]_i}\right) \times \left(\frac{1}{N}\sum_{n=1}^{N} \frac{\partial \varphi_{r}(n)}{\partial [\boldsymbol{\vartheta}]_j}\right) \right\} \notag\\
    =& 2N\rho_1 \mathrm{Cov}\left(\frac{\partial \varphi_{r}}{\partial [\boldsymbol{\vartheta}]_i},\frac{\partial \varphi_{r}}{\partial [\boldsymbol{\vartheta}]_j}\right),
\end{align}
where $\mathbf{J}_{(i, \varphi_{\alpha})} = 2\rho_1 \sum_{n=1}^{N} \frac{\partial \varphi_{r}(n)}{\partial [\boldsymbol{\vartheta}]_i}$ and $\mathbf{J}_{(\varphi_{\alpha},\varphi_{\alpha})} = 2N\rho_1$.

Applying the continuous aperture approximation to the distance $r_{(n_{\mathrm{x}}, n_{\mathrm{y}})}$ in \eqref{eq:r_n} yields
\begin{align}
    r(x, y) \approx r - (x u_{\mathrm{x}} + y u_{\mathrm{y}}) + \frac{x^2(1-u_{\mathrm{x}}^2) + y^2(1-u_{\mathrm{y}}^2)}{2r},
\end{align}
where $x = n_{\mathrm{x}}d \in [-\frac{L_{\mathrm{x}}}{2}, \frac{L_{\mathrm{x}}}{2}]$, $y = n_{\mathrm{y}}d \in [-\frac{L_{\mathrm{y}}}{2}, \frac{L_{\mathrm{y}}}{2}]$, $L_{\mathrm{x}}$ and $L_{\mathrm{y}}$ denote the array apertures along the long and short axes, respectively.
Thus, we have $\varphi_{r}(x, y) = -\frac{2\pi}{\lambda}r(x, y)$, and the first-order partial derivatives of $\varphi_{r}(x, y)$ with respect to $r$, $u_{\mathrm{x}}$, and $u_{\mathrm{y}}$ are derived as
\begin{subequations}
\begin{align}
    \frac{\partial \varphi_{r}}{\partial r} &= -\frac{2\pi}{\lambda} + \frac{\pi}{\lambda r^2} \left\{ x^2(1-u_{\mathrm{x}}^2) + y^2(1-u_{\mathrm{y}}^2) \right\}, \label{eq:partial_r} \\
    \frac{\partial \varphi_{r}}{\partial u_{\mathrm{x}}} &= \frac{2\pi}{\lambda} \left( x + \frac{x^2 u_{\mathrm{x}}}{r} \right), \label{eq:partial_ux} \\
    \frac{\partial \varphi_{r}}{\partial u_{\mathrm{y}}} &= \frac{2\pi}{\lambda} \left( y + \frac{y^2 u_{\mathrm{y}}}{r} \right). \label{eq:partial_uy}
\end{align}
\end{subequations}
Considering $x$ and $y$ as independent, uniformly distributed variables over the array aperture, we have $\mathbb{E}(x) = \mathbb{E}(x^3) = 0$, $\mathbb{E}(y) = \mathbb{E}(y^3) = 0$, $\mathrm{Var}(x) = \frac{L_{\mathrm{x}}^2}{12}$, $\mathrm{Var}(x^2) = \frac{L_{\mathrm{x}}^4}{180}$, $\mathrm{Var}(y) = \frac{L_{\mathrm{y}}^2}{12}$, and $\mathrm{Var}(y^2) = \frac{L_{\mathrm{y}}^4}{180}$. 
Consequently, the elements of $\tilde{\mathbf{J}}$ can be computed as
\begin{subequations}
\begin{align}
    \tilde{\mathbf{J}}_{(u_{\mathrm{x}},u_{\mathrm{y}})} =& 0, \\
    \tilde{\mathbf{J}}_{(r,r)} =& 2N\rho_1 \mathrm{Var}\left( \frac{\partial \varphi_{r}}{\partial r} \right) \notag\\
    =& 2N\rho_1 \left(\frac{\pi}{\lambda r^2}\right)^2 \left[ (1-u_{\mathrm{x}}^2)^2\frac{L_{\mathrm{x}}^4}{180} + (1-u_{\mathrm{y}}^2)^2\frac{L_{\mathrm{y}}^4}{180} \right], \\
    \tilde{\mathbf{J}}_{(r,u_{\mathrm{x}})} =& 2N\rho_1 \mathrm{Cov}\left( \frac{\partial \varphi_{r}}{\partial r}, \frac{\partial \varphi_{r}}{\partial u_{\mathrm{x}}} \right) \notag\\
    =& 2N\rho_1 \frac{2\pi^2}{\lambda^2 r^3} u_{\mathrm{x}}(1-u_{\mathrm{x}}^2) \frac{L_{\mathrm{x}}^4}{180}, \\
    \tilde{\mathbf{J}}_{(u_{\mathrm{x}},u_{\mathrm{x}})} =& 2N\rho_1 \mathrm{Var}\left( \frac{\partial \varphi_{r}}{\partial u_{\mathrm{x}}} \right) \notag\\
    =& 2N\rho_1 \left(\frac{2\pi}{\lambda}\right)^2 \left( \frac{L_{\mathrm{x}}^2}{12} + \frac{u_{\mathrm{x}}^2}{r^2} \frac{L_{\mathrm{x}}^4}{180} \right).
\end{align}
\end{subequations}
By applying the Schur complement, the Fisher information for $r$, i.e., $J_{\mathrm{r}}$, is obtained as \eqref{eq:J_r}, shown at the top of the next page.
\begin{figure*}
\begin{align}\label{eq:J_r}
    J_{\mathrm{r}} =& \tilde{\mathbf{J}}_{(r,r)} - 
    \begin{bmatrix} \tilde{\mathbf{J}}_{(r,u_{\mathrm{x}})} & \tilde{\mathbf{J}}_{(r,u_{\mathrm{y}})} \end{bmatrix}
    \begin{bmatrix} \tilde{\mathbf{J}}_{(u_{\mathrm{x}},u_{\mathrm{x}})} & 0 \\ 0 & \tilde{\mathbf{J}}_{(u_{\mathrm{y}},u_{\mathrm{y}})} \end{bmatrix}^{-1}
    \begin{bmatrix} \tilde{\mathbf{J}}_{(r,u_{\mathrm{x}})} \\ \tilde{\mathbf{J}}_{(r,u_{\mathrm{y}})} \end{bmatrix} \notag\\
    =& \tilde{\mathbf{J}}_{(r,r)} - \frac{\tilde{\mathbf{J}}_{(r,u_{\mathrm{x}})}^2}{\tilde{\mathbf{J}}_{(u_{\mathrm{x}},u_{\mathrm{x}})}} - \frac{\tilde{\mathbf{J}}_{(r,u_{\mathrm{y}})}^2}{\tilde{\mathbf{J}}_{(u_{\mathrm{y}},u_{\mathrm{y}})}} \notag\\
    =& \frac{2N\rho_1\pi^2}{\lambda^2 r^4} \left\{ \frac{(1-u_{\mathrm{x}}^2)^2 L_{\mathrm{x}}^4}{180} \frac{1}{1 + \frac{u_{\mathrm{x}}^2 L_{\mathrm{x}}^2}{15 r^2}} + \frac{(1-u_{\mathrm{y}}^2)^2 L_{\mathrm{y}}^4}{180} \frac{1}{1 + \frac{u_{\mathrm{y}}^2 L_{\mathrm{y}}^2}{15 r^2}} \right\}
\end{align}
{\noindent}\rule[-5pt]{17.5cm}{0.05em}
\end{figure*}
Finally, the CRB for distance estimation is given by the inverse of $J_{\mathrm{r}}$, shown as \eqref{eq:CRB_r}, which completes the proof.

\section{Proof of Theorem \ref{the:PEB}}\label{proof:the:PEB}
By constructing the Jacobian matrix $\mathbf{H} \triangleq \frac{\partial \mathbf{p}_{\mathrm{u}}}{\partial \tilde{\boldsymbol{\vartheta}}}$ that maps the estimated parameter vector $\tilde{\boldsymbol{\vartheta}} = [r, u_{\mathrm{x}}, u_{\mathrm{y}}]^T$ to the position vector $\mathbf{p}_{\mathrm{u}} = [x, y, z]^T$, the CRB for the estimated parameter can be transformed into the PEB, as follows \cite{Abu-Shaban-18TWC}
\begin{align}\label{eq:PEB_proof}
    \mathrm{PEB}^2 =& \mathrm{tr}\left(\mathbf{H} \mathrm{CRB}_{\tilde{\boldsymbol{\vartheta}}} \mathbf{H}^T \right) \notag\\
    =& \mathrm{CRB}_{r} + \frac{r^2}{u_{\mathrm{z}}^2} \{ (1-u_{\mathrm{y}}^2)\mathrm{CRB}_{u_{\mathrm{x}}}  \notag\\
    & + (1-u_{\mathrm{x}}^2)\mathrm{CRB}_{u_{\mathrm{y}}} + 2u_{\mathrm{x}}u_{\mathrm{y}}\mathrm{Cov}(u_{\mathrm{x}}, u_{\mathrm{y}})\},
\end{align}
where $u_{\mathrm{z}} \triangleq \sqrt{1 - u_{\mathrm{x}}^2 - u_{\mathrm{y}}^2}$, $\mathrm{CRB}_{u_{\mathrm{x}}}$ and $\mathrm{CRB}_{u_{\mathrm{y}}}$ represent the CRBs of $u_{\mathrm{x}}$ and $u_{\mathrm{y}}$, respectively.
Similar to the derivation of $\mathrm{CRB}_{r}$ in \eqref{eq:CRB_r}, we obtain
\begin{subequations}\label{eq:CRB_xy}
\begin{align}
    \mathrm{CRB}_{u_{\mathrm{x}}} =& \left(\tilde{\mathbf{J}}_{(u_{\mathrm{x}},u_{\mathrm{x}})} - \frac{\tilde{\mathbf{J}}_{(r,u_{\mathrm{x}})}^2}{\tilde{\mathbf{J}}_{(r,r)} - \frac{\tilde{\mathbf{J}}_{(r,u_{\mathrm{y}})}^2}{\tilde{\mathbf{J}}_{(u_{\mathrm{y}},u_{\mathrm{y}})}}}\right)^{-1}, \\
    \mathrm{CRB}_{u_{\mathrm{y}}} =& \left(\tilde{\mathbf{J}}_{(u_{\mathrm{y}},u_{\mathrm{y}})} - \frac{\tilde{\mathbf{J}}_{(r,u_{\mathrm{y}})}^2}{\tilde{\mathbf{J}}_{(r,r)} - \frac{\tilde{\mathbf{J}}_{(r,u_{\mathrm{x}})}^2}{\tilde{\mathbf{J}}_{(u_{\mathrm{x}},u_{\mathrm{x}})}}}\right)^{-1},
\end{align}
\end{subequations}
where $\tilde{\mathbf{J}}_{(u_{\mathrm{y}}, u_{\mathrm{y}})} = 2N\rho_1 (\frac{2\pi}{\lambda})^2 \times ( \frac{L_{\mathrm{y}}^2}{12} + \frac{u_{\mathrm{y}}^2}{r^2} \frac{L_{\mathrm{y}}^4}{180} )$.
Along the broadside direction, where $u_{\mathrm{x}} = 0$, $u_{\mathrm{y}} = 0$, and $u_{\mathrm{z}} = 1$, \eqref{eq:PEB_proof} simplifies to
\begin{align}
    \mathrm{PEB}^2 =& \mathrm{CRB}_{r} + r^2 \mathrm{CRB}_{u_{\mathrm{x}}} + r^2 \mathrm{CRB}_{u_{\mathrm{y}}} \notag\\
    =& \frac{90 \lambda^2 r^4}{\pi^2 d^4 \rho_1 N^3 (\gamma^2 + \gamma^{-2})} + \frac{3 \lambda^2 r^2}{2\pi^2 \rho_1 N^2 d^2} (\gamma+\gamma^{-1}).
\end{align}
The first derivative of $\mathrm{PEB}^2$ with respect to $\gamma$ is computed as
\begin{align}
    \frac{\partial \mathrm{PEB}^2}{\partial \gamma} = \left(1 - \gamma^{-2}\right) \times \left\{ \kappa_2 - \kappa_1 \frac{2\gamma(1+\gamma^{-2})}{\left(\gamma^2 + \gamma^{-2}\right)^2} \right\},
\end{align}
where $\kappa_1 \triangleq \frac{90\lambda^2 r^4}{\pi^2 d^4 \rho_1 N^3}$ and $\kappa_2 \triangleq \frac{3\lambda^2 r^2}{2\pi^2\rho_1 N^2 d^2}$.
By setting $\frac{\partial \mathrm{PEB}^2}{\partial \gamma} = 0$, we obtain two stationary points, namely $\gamma_1$ and $\gamma_2 = 1$, where $\gamma_1$ is the root of the equation $\frac{2\gamma(1+\gamma^{-2})}{(\gamma^2 + \gamma^{-2})^2} = \frac{\kappa_2}{\kappa_1}$.

By defining $f(\gamma) \triangleq \frac{2(\gamma + \gamma^{-1})}{(\gamma^2 + \gamma^{-2})^2}$, the first derivative of $f(\gamma)$ with respect to $\gamma$ can be written as
\begin{align}
    \frac{\partial f(\gamma)}{\partial \gamma} =& \frac{2(1 - \gamma^{-2})(\gamma^2 + \gamma^{-2})^2}{(\gamma^2 + \gamma^{-2})^4} \notag\\
    & - \frac{4(\gamma + \gamma^{-1}) (\gamma^2 + \gamma^{-2})(2\gamma - 2\gamma^{-3})}{(\gamma^2 + \gamma^{-2})^4} \notag\\
    =& \frac{-6(\gamma^2 - \gamma^{-4}) - 10(1 - \gamma^{-2})}{(\gamma^2 + \gamma^{-2})^3}.
\end{align}
For $\gamma > 1$, it is straightforward to verify that $\frac{\partial f(\gamma)}{\partial \gamma} < 0$, meaning that $f(\gamma)$ is monotonically decreasing within the interval $\gamma \in (1, +\infty)$, with $f(1) = 1$ and $\lim_{\gamma \to \infty} f(\gamma) = 0$.
Given that $\frac{\kappa_2}{\kappa_1} = \frac{N d^2}{60 r^2}$, we have $\frac{\kappa_2}{\kappa_1} \in (0, 1)$ for any distance $r > r_{\mathrm{th}}$, where $r_{\mathrm{th}} \triangleq \sqrt{\frac{N d^2}{60}}$.
Therefore, the equation $f(\gamma) = \frac{\kappa_2}{\kappa_1}$ is guaranteed to possess a unique real root $\gamma_1 > 1$.
Since the derivative $\frac{\partial \mathrm{PEB}^2}{\partial \gamma}$ changes sign from negative to positive at $\gamma_1$, this stationary point corresponds to the unique minimum point for $r > r_{\mathrm{th}}$.

To find the asymptotic closed-form expression for $\gamma_1$ when $\gamma \gg 1$, we introduce the infinitesimal variable $\epsilon = \gamma^{-1}$.
By applying the Maclaurin expansion, $f(\epsilon^{-1})$ can be rewritten as
\begin{align}\label{eq:epsilon}
    f(\epsilon^{-1}) =& \frac{2\epsilon^3(1 + \epsilon^2)}{(1 + \epsilon^4)^2} \notag\\
    =& 2\epsilon^3 + 2\epsilon^5 + \mathcal{O}(\epsilon^7).
\end{align}
Substituting \eqref{eq:epsilon} into the equation $f(\gamma) = \frac{\kappa_2}{\kappa_1}$ yields
\begin{align}
    \epsilon = \beta^{\frac{1}{3}} + \mathcal{O}\left(\beta^{\frac{1}{3}}\right).
\end{align}
where $\beta \triangleq \frac{N d^2}{120 r^2}$.
Then, $\gamma_1$ converges to $\gamma_1 \approx \beta^{-\frac{1}{3}}$.

Next, we evaluate the condition $r \leq r_{\mathrm{th}}$, where $\frac{\partial \mathrm{PEB}^2}{\partial \gamma} \geq 0$ holds for all $\gamma \in [1, +\infty)$.
This monotonic property makes $\gamma_2 = 1$ the unique minimum point for $\mathrm{PEB}^2$ in this regime.
Notably, $r_{\mathrm{th}}$ is typically extremely small in practical deployment scenarios.
For instance, given $N = 2048$ and a carrier frequency of $28\,\mathrm{GHz}$, we have $r_{\mathrm{th}} = 0.031\,\mathrm{m}$, which falls well below the practical operating distance of communication systems and can be neglected.
Consequently, $\gamma_1$ in \eqref{eq:gamma_opt_asy} stands as the unique optimal solution that minimizes the $\mathrm{PEB}^2$, which completes the proof.

\footnotesize
\bibliographystyle{IEEEtran}
\bibliography{reference}
\end{document}